\documentclass[english,runningheads,11pt]{llncs}
\usepackage{amsmath,amssymb,graphicx,mathrsfs}
\usepackage{tabularx,booktabs,multirow,delarray,array}
\usepackage{latexsym}
\usepackage{epstopdf}
\usepackage[linesnumbered, vlined, ruled]{algorithm2e}

\usepackage[in]{fullpage}
\def\calC{\mathcal{C}}
\def\calI{\mathcal{I}}
\def\calL{\mathcal{L}}

\newtheorem{observation}{Observation}

\begin{document}

\title{Separating Overlapped Intervals on a Line\thanks{This research was supported in part by NSF under Grant CCF-1317143.}
}

\author{
Shimin Li
\and
Haitao Wang
}

\institute{
Department of Computer Science\\
Utah State University, Logan, UT 84322, USA\\
\email{shiminli@aggiemail.usu.edu, haitao.wang@usu.edu}\\
}

\maketitle

\pagestyle{plain}
\pagenumbering{arabic}
\setcounter{page}{1}

\begin{abstract}
%In this paper, we studied the interval spreading problem.
Given $n$ intervals on a line $\ell$,
we consider the problem of moving these intervals on $\ell$
such that after the movement no two intervals overlap and the maximum moving distance of
the intervals is minimized.
%We present an $O(n\log n)$ time algorithm for it.
The difficulty for solving the
problem lies in determining the order of the intervals in an optimal
solution. By interesting observations, we show that it is sufficient
to consider at most $n$ ``candidate'' lists of ordered intervals.
Further, although explicitly maintaining these lists takes
$\Omega(n^2)$ time and space, by more observations and
a pruning technique, we present an algorithm that can compute an optimal solution in
$O(n\log n)$ time and $O(n)$ space. We also prove an $\Omega(n\log n)$ time lower
bound for solving the problem, which implies the optimality of our algorithm.
\end{abstract}

\section{Introduction}
\label{sec:intro}
Let $\calI$ be a set of $n$ intervals on a real line $\ell$.
We say that two intervals {\em overlap} if their intersection contains
more than one point. In this paper, we consider an {\em
interval separation problem}: move the intervals of $\calI$ on $\ell$
such that
no two intervals overlap and the maximum moving distance of these intervals is minimized.
%Note that the intervals of $\calI$ may have different lengths.

If all intervals of $\calI$ have the same length, then after the left
endpoints of the intervals are sorted, the problem can be solved
in $O(n)$ time by an easy greedy algorithm \cite{ref:LiAl15}. For the
general problem where intervals may have different
lengths, to the best of our knowledge, the problem has not been
studied before. In this paper, we present an $O(n\log n)$ time and
$O(n)$ space algorithm for it. We also show an $\Omega(n\log n)$ time lower
bound for solving the problem under the algebraic decision tree model, and thus
our algorithm is optimal.

As a basic problem and
like many other interval problems, the interval separation problem
potentially has many applications. For example, one possible application is on scheduling, as follows.
Suppose there are $n$ jobs that need to be completed on a machine. Each job
requests a starting time and a total time for using the machine
(hence it is a time interval). The machine can only
work on one job at any time, and once it works on one job, it is not
allowed to switch to other jobs until the job is finished. If the
requested time intervals of the jobs have any overlap, then we have to
change the requested starting times of some intervals. In order to
minimize deviations from their requested time intervals, one scheduling
strategy could be changing the requested starting times (either advance
or delay) such that the maximum difference between the requested
starting times and the scheduled starting times of all jobs is minimized.
Clearly, the problem is an instance of the interval separation
problem.
The problem also has applications in the following
scenario. Suppose a wireless sensor network has $n$ wireless mobile devices on a line and each
device has a transmission range. We want to move the devices along the
line to eliminate the interference such that the maximum moving
distance of the devices is minimized (e.g., to save the energy).
This is also an instance of the interval separation problem.

\subsection{Related Work}
Many interval problems have been used to model scheduling problems. We
give a few examples. Given $n$ jobs, each job requests a time
interval to use a machine. Suppose there is only one machine and the goal
is to find a maximum number of jobs whose requested time intervals do
not have any overlap (so that they can use the machine). The problem
can be solved in $O(n\log n)$ time by an easy greedy
algorithm~\cite{ref:KleinbergAl05Ch4}. Another related problem is to find
a minimum number of machines such that all jobs can be
completed~\cite{ref:KleinbergAl05Ch4}.
%This problem can also be solved efficiently~\cite{ref:KleinbergAl05Ch4}.
Garey et al.~\cite{ref:GareySc81} studied a scheduling problem, which is essentially
the following problem. Given $n$ intervals on a line, determine whether it is
possible to find a unit-length sub-interval in each input interval,
such that no two sub-intervals overlap.
An $O(n\log n)$ time algorithm was given in \cite{ref:GareySc81} for
it. An optimization version of the problem was also
studied \cite{ref:ChrobakOn07,ref:VakhaniaA13},
where the goal is to find a maximum number of intervals that contain non-overlapping
unit-length sub-intervals.
%Chrobak et al.~\cite{ref:ChrobakOn07}
%gave an $O(n^5)$ time algorithm for it, which was later improved to $O(n^2\log n)$
%time~\cite{ref:VakhaniaA13}.
%by Vakhania~\cite{ref:VakhaniaA13}.
%The online version of the problem was also considered~\cite{ref:ChrobakA06}.
Other scheduling problems on intervals have also been considered,
e.g., see \cite{ref:ChrobakA06,ref:GareySc81,ref:KleinbergAl05Ch4,ref:LangSc76,ref:LawlerSe93,ref:SimonsA78,ref:VakhaniaMi13}.

Many problems on wireless sensor networks are also modeled as interval
problems. For example, a mobile sensor barrier coverage problem can be modeled as
the following interval problem. Given on a line $n$ intervals (each interval is
the region covered by a sensor at the center of the
interval) and another segment $B$ (called ``barrier''), the
goal is to move the intervals such that the union of the intervals
fully covers $B$ and the maximum moving distance of all intervals is
minimized. If all intervals have the same length, Czyzowicz et
al.~\cite{ref:CzyzowiczOn09} solved the problem in $O(n^2)$ time and
later Chen et al.~\cite{ref:ChenAl13} improved it to $O(n\log n)$ time.
If intervals have different lengths, Chen et
al.~\cite{ref:ChenAl13} solved the problem in $O(n^2\log n)$ time. The
min-sum version of the problem has also been considered. If
intervals have the same length, Czyzowicz et al.~\cite{ref:CzyzowiczOn10}
gave an $O(n^2)$ time algorithm, and Andrews and
Wang~\cite{ref:AndrewsMi15} solved the problem in $O(n\log n)$ time.
If intervals have different lengths, then the problem becomes
NP-hard~\cite{ref:ChenAl13}. Refer
to~\cite{ref:Bar-NoyMa13,ref:BhattacharyaOp09,ref:ChenOp15,ref:LiMi11,ref:MehrandishOn11,ref:MehrandishMi11}
for other interval problems on mobile sensor barrier coverage.

Our interval separation problem may also be considered as a
coverage problem in the sense that we want to move intervals of
$\calI$ to cover a total of maximum length
of the line $\ell$ such that the maximum moving distance of the
intervals is minimized.

\subsection{Our Approach}
\label{sec:approach}
We consider a {\em one-direction} version of the problem in which
intervals of $\calI$ are only allowed to move rightwards. We show (in Section~\ref{sec:pre}) that
the original ``two-direction'' problem can be reduced to the
one-direction problem in the following way: If OPT is an optimal
solution of the one-direction problem and $\delta_{opt}$ is the
maximum moving distance of all intervals in OPT, then we can obtain an
optimal solution for the two-direction problem by moving each interval
in OPT leftwards by $\delta_{opt}/2$.

Hence, it is sufficient to solve the one-direction problem. It turns
out that the difficulty is mainly on determining the order of intervals
of $\calI$ in OPT.
Indeed, once such an ``optimal order'' is known,
it is quite straightforward to compute the positions of the
intervals in OPT in additional $O(n)$ time
(i.e., consider the intervals in the order one by
one and put each interval ``as left as possible'').
If all intervals have the same length, then such an
optimal order is obvious, which is the order of the intervals sorted
by their left endpoints in the input. Indeed, this is how the $O(n)$
time algorithm in \cite{ref:LiAl15} works.

However, if the intervals have
different lengths, which is the case we consider in this paper, then determining
an optimal order is substantially more challenging. At first glance, it seems
that we have to consider all possible orders of the intervals, whose
number is exponential. By several interesting (and even surprising)
observations, we show that we only need to consider at most $n$
ordered lists of intervals. Consequently, a straightforward algorithm
can find and maintain these ``candidate'' lists in $O(n^2)$ time
and space.  We call it the ``preliminary algorithm'', which is
essentially a greedy algorithm. The algorithm is
relatively simple but it is quite involved to prove its correctness.
To this end, we extensively use the ``exchange argument'', which is a standard
technique for proving correctness of greedy algorithms (e.g.,
see~\cite{ref:KleinbergAl05Ch4}).

To further improve the preliminary algorithm, we discover more
observations, which help us ``prune'' some ``redundant'' candidate lists. More
importantly, the remaining lists have certain monotonicity properties
such that we are able to implicitly
compute and maintain them in $O(n\log n)$ time and $O(n)$
space, although the number of the lists can still be $\Omega(n)$.
Although the correctness analysis is fairly complicated, the
algorithm is still quite simple and easy to implement
(indeed, the most ``complicated'' data structure is a binary search tree).

The rest of the paper is organized as follows. In Section~\ref{sec:pre}, we give notation
and reduce our problem to the one-direction case. In
Section~\ref{sec:prealgo}, we give our
preliminary algorithm, whose correctness is proved in
Section~\ref{sec:preproof}. The improved algorithm is presented in
Section~\ref{sec:improve}. In Section~\ref{sec:conclude}, we conclude the paper and
prove the $\Omega(n\log n)$ time lower bound by a reduction from the integer
element distinctness problem  \cite{ref:LubiwA91,ref:YaoLo91}.

\section{Preliminaries}
\label{sec:pre}

We assume the line $\ell$ is the $x$-axis.
The {\em one-direction} version of the interval separation problem is
to move intervals of $\calI$ on $\ell$ in one direction (without
loss of generality, we assume it is the right direction) such
that no two intervals overlap and the maximum moving distance of the
intervals is minimized. Let OPT denote an optimal solution of the
one-direction version and let $\delta_{opt}$ be the maximum moving
distance of all intervals in OPT. The following lemma gives a
reduction from the general ``two-direction'' problem to the
one-direction problem.

\begin{lemma}\label{lem:reduction}
An optimal solution for the interval separation problem can be
obtained by moving every interval in OPT leftwards by $\delta_{opt}/2$.
\end{lemma}
\begin{proof}
Let $SOL$ be the solution obtained
by moving every interval in OPT leftwards by $\delta_{opt}/2$.
Our goal is to show that $SOL$ is an optimal solution for our original problem. Let
$\delta$ be the maximum moving distance of all intervals in $SOL$. Since
no intervals in OPT have been moved leftwards (with respect to their
input positions), we have
$\delta=\delta_{opt}/2$.

Assume to the contrary that $SOL$ is not optimal. Then, there exists
another solution $SOL'$ for the original problem
in which the maximum interval moving distance is
$\delta'<\delta$. By moving every interval of $SOL'$ rightwards by
$\delta'$, we can obtain a feasible solution $SOL''$ for the
one-direction problem in which no
interval has been moved leftwards (with respect to their input
positions) and the maximum interval moving distance of $SOL''$ is at
most $2\delta'$, which is smaller than $\delta_{opt}$ since $\delta'<\delta$.
However, this contradicts with that OPT is an optimal solution for the
one-direction case.
\qed
\end{proof}

By Lemma~\ref{lem:reduction}, once we have an optimal solution for the
one-direction problem, we can obtain an optimal solution
for our original problem in additional $O(n)$ time. In the following,
we will focus on solving the one-direction case.

We first sort all intervals of $\calI$ by their left endpoints.
For ease of exposition, we assume no two intervals have their left
endpoints located at the same position (otherwise we
could break ties by also sorting their right endpoints). Let
$\calI=\{I_1,I_2,\ldots,I_n\}$ be the sorted intervals by their
left endpoints from left to right.
For each (integer) $i\in [1,n]$, denote by $l_i$ and $r_i$ the (physical) left and right
endpoints of $I_i$, respectively. Denote by $x_i^l$ and $x_i^r$ the
$x$-coordinates of $l_i$ and $r_i$ in the input, respectively. Note that for each $i\in [1,n]$, the two physical endpoints $l_i$ and $r_i$ may be moved during the algorithm, but the two coordinates  $x_i^l$ and $x_i^r$ are always fixed.
Denote by $|I_i|$ the length of $I_i$, i.e., $|I_i|=x_i^r-x_i^l$.

For convenience, when we say the {\em position} of an interval, we refer
to the position of the left endpoint of the interval.

With respect to a subset $\calI'$ of $\calI$, by a {\em configuration}
of $\calI'$, we refer to a specification of the position of each
interval of $\calI'$. For example, in the input configuration of
$\calI$, interval $I_i$ is at  $x_i^l$ for each $i\in [1,n]$. Given a configuration
$\calC$ of $\calI'$, for each
interval $I_i\in \calI'$, if $l_i$ is at $x$ in $\calC$, then we call
the value $x-x_i^l$ the {\em displacement} of $I_i$, denoted by
$d(i,\calC)$, and if $d(i,\calC)\geq 0$, then we say that $I_i$ is {\em valid} in $\calC$.
We say that $\calC$ is {\em feasible} if the displacement of every
interval of $\calI'$ is valid and no two intervals of $\calI'$ overlap
in $\calC$. The maximum displacement of the intervals of $\calI'$ in
$\calC$ is called the {\em max-displacement} of $\calC$, denoted by
$\delta(\calC)$. Hence, finding an optimal solution for the
one-direction problem is equivalent to computing a feasible
configuration of $\calI$ whose max-displacement is minimized; such a
configuration is also called an {\em optimal configuration}.

For convenience of discussion, depending on the context, we will use
the intervals $I_i$ of $\calI$ and their indices $i$ interchangeably.
For example, $\calI$ may also refer to the set of indices
$\{1,2,\ldots,n\}$.

%Consider any optimal configuration OPT.
Let $L_{opt}$ be the list of intervals of $\calI$ in an optimal configuration sorted from left to right. We call $L_{opt}$ an {\em optimal list}.
Given $L_{opt}$, we can compute an optimal configuration in
$O(n)$ time by an easy greedy algorithm, called the {\em left-possible placement strategy}: Consider the intervals
following their order in $L_{opt}$, and for each interval, place it
on $\ell$ as left as possible so that it does not overlap with the
intervals that are already placed on $\ell$ and its displacement is
non-negative. The following lemma formally gives the algorithm and
proves its correctness.

\begin{lemma}\label{lem:left}
Given an optimal list $L_{opt}$, we can compute an optimal
configuration in $O(n)$ time by the left-possible placement strategy.
\end{lemma}
\begin{proof}
We first describe the algorithm and then prove its correctness.

We consider the indices one by one
following their order in $L_{opt}$. Consider any index $i$. If $I_i$ is
the first interval of $L_{opt}$, then we place $I_i$ at $x_i^l$ (i.e.,
$I_i$ stays at its input position).
%we place the left endpoint $l_i$ of $I_i$ at $x_i^l$, i.e., its position in the input.
Otherwise, let $I_j$ be the previous interval of $I_i$ in $L_{opt}$. So $I_j$ has
already been placed on $\ell$. Let $x$ be the current $x$-coordinate
of the right endpoint $r_j$ of $I_j$. We place the left endpoint
$l_i$ of $I_i$ at $\max\{x_i^l,x\}$.
If $I_i$ is the last interval of $L_{opt}$, then we finish the algorithm.
Clearly, the algorithm can be easily implemented in $O(n)$ time.

Let $\calC$ be the configuration of all intervals obtained by the above algorithm.
Recall that $\delta(\calC)$ denote the max-displacement of $\calC$.
Below, we show that $\calC$ is an optimal configuration.

Indeed, since
$L_{opt}$ is an optimal list, there exists an optimal configuration
$\calC'$ in which the order of the indices of $\calI$ follows that in
$L_{opt}$. Hence, the max-displacement of $\calC'$ is $\delta_{opt}$.
According to our greedy strategy for computing $\calC$, it is not difficult to
see that the position of each interval $I_i$ of $\calI$ in $\calC$ cannot be strictly to
the right of its position in $\calC'$. Therefore, the displacement of
each interval in $\calC$ is no larger than that in $\calC'$. This
implies that $\delta(\calC)\leq \delta_{opt}$. Therefore, $\calC$ is
an optimal configuration.
\qed
\end{proof}

Due to Lemma~\ref{lem:left}, we will focus on computing an
optimal list $L_{opt}$.

%With a little abuse of notation, we also use $\calI$ to denote the set
%of indices $\{1,2,\ldots, n\}$ (indeed, each index $i$ refers to the interval $I_i$).

For any subset $\calI'$ of $\calI$, an {\em (ordered) list} of $\calI'$
refers to a permutation of the indices of $\calI'$.
Let $L$ be a list of $\calI$ and let $L'$ be a list of $\calI'$ with
$\calI'\subseteq \calI$. We say that $L'$ is {\em consistent with} $L$
if the relative order of indices of $\calI'$ in $L$ is the same as
that in $L'$. If $L'$ is consistent with an optimal list $L_{opt}$ of
$\calI$, then we call $L'$ a {\em canonical list} of $\calI'$.

For any $1\leq i\leq j\leq n$, we use $\calI[i,j]$ to
denote the subset of consecutive intervals of $\calI$ from $i$ to $j$,
i.e, $\{i,i+1,\ldots,j\}$.
%For any configuration $\calC$

\section{The Preliminary Algorithm}
\label{sec:prealgo}

In this section, we describe an algorithm that can compute an optimal
list in $O(n^2)$ time and space. The correctness of the algorithm is
mainly discussed in Section~\ref{sec:preproof}.

%\subsection{The Algorithm}

Our algorithm considers the intervals of $\calI$ one by one by their
index order. After each interval $I_i$ is processed, we obtain a set
$\calL$ of at most $i$ lists of the indices of $\calI[1,i]$, such that
$\calL$ contains at least one canonical list of $\calI[1,i]$.
For each list $L\in \calL$, a feasible configuration $\calC_L$ of the intervals
of $\calI[1,i]$ is also maintained. As will be clear later,
$\calC_L$ is essentially the configuration obtained by applying the
left-possible placement strategy on the intervals of $\calI[1,i]$
following their order in $L$. For each $j\in [1,i]$, we let $x_j^l(\calC_L)$
and $x_j^r(\calC_L)$ respectively denote the $x$-coordinates of $l_j$
and $r_j$ in $\calC_L$ (recall that $l_j$ and $r_j$ are the left and
right endpoints of the interval $I_j$, respectively). Recall that
$\delta(\calC_L)$ denotes the max-displacement of $\calC_L$, i.e, the maximum displacement of the intervals
of $\calI[1,i]$ in $\calC_L$.

Initially when $i=1$, we have only one list $L=\{1\}$ and let $\calC_L$ consist
of the single interval $I_1$ at its input position, i.e.,
$x_1^l(\calC_L)=x_1^l$. Clearly, $\delta(\calC_L)=0$. We let $\calL$
consist of the only list $L$. It is vacuously true that $L$ is a canonical list of
$\calI[1,1]$.

In general, assume interval $I_{i-1}$ has been processed and we have
the list set $\calL$ as discussed above. In the following, we give our
algorithm for processing $I_{i}$. Consider a list $L\in \calL$. Note
that $\calC_L$ has been computed, which is a feasible configuration of
$\calI[1,i-1]$. The value $\delta(\calC_L)$ is also maintained.
Let $m$ be the last index in $L$. Note that $m<i$.
Depending on the values
of $x_i^l$, $x_i^r$, $x_m^r$, and $x_m^l(\calC_L)$, there are three main cases
(e.g. see Fig.~\ref{fig:maincases}).

\begin{figure}[t]
\begin{minipage}[t]{\textwidth}
\begin{center}
\includegraphics[height=0.8in]{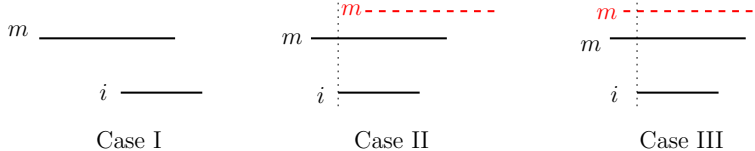}
\caption{\footnotesize Illustrating the three main cases. The (black)
solid segments show intervals in their input positions
and the (red) dashed segments shows interval $I_m$ in $\calC_L$.}
\label{fig:maincases}
\end{center}
\end{minipage}
\vspace{-0.15in}
\end{figure}

\paragraph{Case I: $x_{i}^r \geq x_m^r$ (i.e., the right endpoint $r_i$ of $I_{i}$ is to
the right of $r_m$ in the input).} In this case, we update
$L$ by appending $i$ to the end of $L$. Further, we update the configuration
$\calC_L$ by placing $l_{i}$ at
$\max\{x_m^r(\calC_L),x_{i}^l\}$ (which follows the left-possible
placement strategy). We let $L'$ denote the original list of $L$ before $i$ is
inserted and let $\calC_{L'}$ denote the original configuration of $\calC_L$.
We update $\delta(\calC_L)$ by the following observation.

\begin{observation}\label{obser:case1}
$\calC_L$ is a feasible configuration and
$\delta(\calC_L)=\max\{\delta(\calC_{L'}),x_{i}^l(\calC_{L})-x_{i}^l\}$.
\end{observation}
\begin{proof}
By our way of setting $I_i$ in $\calC_L$, $I_{i}$ is valid and does not overlap with any other interval in
$\calC_L$. Hence, $\calC_L$ is feasible. Comparing with $\calC_{L'}$,
$\calC_L$ has one more interval $I_{i}$. Therefore, $\delta(\calC_L)$
is equal to the larger value of $\delta(\calC_{L'})$ and the displacement
of $I_{i}$ in $\calC_L$, which is $x_{i}^l(\calC_L)-x_{i}^l$.
\qed
\end{proof}

The following lemma will be used to show the correctness of our
algorithm and its proof is deferred to Section~\ref{sec:preproof}.
\begin{lemma}\label{lem:case1}
If $L'$ is a canonical list of $\calI[1,i-1]$, then $L$ is a canonical
list of $\calI[1,i]$.
%and further, if $i+1=n$, $\calC_L$ is an optimal configuration.
\end{lemma}

\paragraph{Case II: $x_{i}^r < x_m^r$ and $x_{i}^l\leq x_m^l(\calC_L)$.} In this case,
we update $L$ by inserting $i$ right before $m$. Let $x=x_m^l(\calC_L)$.
We update $\calC_L$ by setting $l_{i}$ at $x$ and setting
$l_m$ at $x+|I_{i}|$. We let $L'$ denote the original list of $L$ before inserting $i$ and let
$\calC_{L'}$ denote the original $\calC_L$.
We update $\delta(\calC_L)$ by the following observation.
Note that $x_m^l(\calC_L)$ now refers to
the position of $l_m$ in the updated $\calC_L$.

\begin{observation}\label{obser:case2}
$\calC_L$ is a feasible configuration and
$\delta(\calC_L)=\max\{\delta(\calC_{L'}),x_m^l(\calC_L)-x_m^l\}$.
\end{observation}
\begin{proof}
%Since $\calC_{L'}$ is feasible, $l_{i}$ is at $x$, $l_m$ is at
%$x+len(I_{i})$, no two intervals of $\calI[1,i]$ overlap in $\calC_L$.
Since $x_{i}^l\leq x$ and $l_{i}$ is at
$x$ in $\calC_L$, $I_{i}$ is valid in $\calC_L$.
Comparing with its position in $\calC_{L'}$, $I_m$ has been moved
rightwards; since $I_m$ is valid in $\calC_{L'}$,
$I_m$ is also valid in $\calC_L$. Note that no two intervals
overlap in $\calC_L$.  Therefore, $\calC_L$ is a feasible
configuration.

Comparing with $\calC_{L'}$, $\calC_L$ has one more interval $I_{i}$ and
$I_m$ has been moved rightwards in $\calC_L$. Therefore, $\delta(\calC_L)$
is equal to the maximum of the following three values:
$\delta(\calC_{L'})$, the displacement of $I_{i}$ in $\calC_L$, and
the displacement of $I_{m}$ in $\calC_L$. Observe that the
displacement of $I_{i}$ is smaller than that of  $I_{m}$. This is
because $l_{m}$ is to the left of $l_{i}$ in the input (since
$m<i$) while $l_m$ is to the right of $l_{i}$ in $\calC_L$. Thus, it holds that
$\delta(\calC_L)=\max\{\delta(\calC_{L'}),x_m^l(\calC_L)-x_m^l\}$.
\qed
\end{proof}

The proof of the following lemma is deferred to Section~\ref{sec:preproof}.
\begin{lemma}\label{lem:case2}
If $L'$ is a canonical list of $\calI[1,i-1]$, then $L$ is a canonical list of $\calI[1,i]$.
\end{lemma}

\paragraph{Case III: $x_{i}^r < x_m^r$ and $x_{i}^l> x_m^l(\calC_L)$.} In this case,
we first update $L$ by appending $i$ to the end of $L$ and update
$\calC_L$ by placing the left endpoint of $I_{i}$ at
$x_m^r(\calC_L)$. Let $L'$ be the original list $L$ before we insert
$i$ and let $\calC_{L'}$ be the original configuration of $\calC_L$.

Further, we create a new list $L^*$, which is the same as $L$
except that we switch the order of $i$ and $m$. Thus, $m$ is the last
index of $L^*$. Correspondingly, the configuration $\calC_{L^*}$ is
the same as $\calC_L$ except that $l_{i}$ is at $x_{i}^l$, i.e., its
position in the input, and $l_m$ is at $x_{i}^r$. We say that $L^*$ is the
{\em new list generated} by $L'$.
We do not put $L^*$ in the set $\calL$ at this moment (but $L$ is in
$\calL$).

\begin{observation}\label{obser:case3}
Both $\calC_L$ and $\calC_{L^*}$ are feasible;
$\delta(\calC_L)=\max\{\delta(\calC_{L'}),x_{i}^l(\calC_L)-x_{i}^l\}$
and
$\delta(\calC_{L^*})=\max\{\delta(\calC_{L'}),x_{m}^l(\calC_{L^*})-x_{m}^l\}$.
\end{observation}
\begin{proof}
By a similar argument as in Observation~\ref{obser:case1}, $\calC_L$ is
feasible and
$\delta(\calC_L)=\max\{\delta(\calC_{L'}),x_{i}^l(\calC_L)-x_{i}^l\}$.
By a similar argument as in Observation~\ref{obser:case2}, $\calC_{L^*}$ is
feasible and
$\delta(\calC_{L^*})=\max\{\delta(\calC_{L'}),x_{m}^l(\calC_{L^*})-x_{m}^l\}$.
We omit the details.
\qed
\end{proof}

The proof of the following lemma is deferred to Section~\ref{sec:preproof}.
\begin{lemma}\label{lem:case3}
If $L'$ is a canonical list of $\calI[1,i-1]$, then one of $L$ and
$L^*$ is a canonical list of $\calI[1,i]$.
%and further, if $i=n$, then
%one of $\calC_L$ and $\calC_{L^*}$ is an optimal configuration.
\end{lemma}

After each list $L$ of $\calL$ is processed as above, let $\calL^*$
denote the set of all new generated lists in Case III. Recall that no list
of $\calL^*$ has been added into $\calL$ yet. Let $L^*_{min}$ be the list of
$\calL^*$ with the minimum value $\delta(\calC_{L^*_{min}})$.
The proof of the following lemma is deferred to Section~\ref{sec:preproof}.
\begin{lemma}\label{lem:prune}
If $\calL^*$ has a canonical list of $\calI[1,i]$, then $L^*_{min}$
is a canonical list of $\calI[1,i]$.
\end{lemma}

Due to Lemma~\ref{lem:prune}, among all lists of $\calL^*$, we only
need to keep $L_{min}^*$.  So we add $L^*_{min}$ to
$\calL$ and ignore all other lists of $\calL^*$. We call $L^*_{min}$ a
{\em new list} of $\calL$ produced by our algorithm for processing
$I_{i}$ and all other lists of $\calL$ are considered as the {\em old
lists}.

\paragraph{Remark.}
Lemma~\ref{lem:prune} is a key observation that helps avoid
maintaining an exponential number of lists.

This finishes our algorithm for processing the
interval $I_{i}$. Clearly, $\calL$ has at most one more new list.
After $I_n$ is processed, the list $L$ of $\calL$ with minimum
$\delta(\calC_L)$ is an optimal list.

According to our above description, the algorithm can be easily
implemented in $O(n^2)$ time and space. The proof of
Theorem~\ref{theo:pre} gives the details and also shows the correctness
of the algorithm based on Lemmas~\ref{lem:case1},
\ref{lem:case2}, \ref{lem:case3}, and \ref{lem:prune}.

\begin{theorem}\label{theo:pre}
An optimal solution for the one-direction problem can be found in
$O(n^2)$ time and space.
\end{theorem}
\begin{proof}
To implement the algorithm, we can use a linked list to represent
each list of $\calL$. Consider a general step for processing interval
$I_{i}$.

For any list $L\in \calL$, inserting
$i$ to $L$ can be easily done in $O(1)$ time for each of the
three cases. The configuration $\calC_L$ and the value $\delta(\calC_L)$
can also be updated in $O(1)$ time.
If $L$ generates a new list $L^*$, then we do not
explicitly construct $L^*$ but only compute the value $\delta(\calC_{L^*})$,
which can be done in $O(1)$ time by Observation~\ref{obser:case3}.
Once every list $L\in \calL$ has been processed, we find the list
$L^*_{min}\in \calL^*$. Then, we
explicitly construct $L^*$ and $\calC_{L^*}$, in
$O(n)$ time.

Hence, each general step for processing $I_{i}$ can be done in
$O(n)$ time since $\calL$ has at most $n$ lists. Thus, the total time
and space of the algorithm is $O(n^2)$.

For the correctness,
after a general step for processing $I_{i}$, Lemmas~\ref{lem:case1},
\ref{lem:case2}, \ref{lem:case3}, and \ref{lem:prune} together guarantee
that the set $\calL$ has at least one canonical list of $\calI[1,i]$. After $I_n$ is processed, since $\calC_L$ is essentially obtained by the left-possible placement strategy for each list $L\in \calL$, if
$L$ is the list of $\calL$ with the smallest $\delta(\calC_L)$, then $L$ is an optimal list and
$\calC_L$ is an optimal configuration by Lemma~\ref{lem:left}.
\qed
\end{proof}

\section{The Correctness of the Preliminary Algorithm}
\label{sec:preproof}

In this section, we establish the correctness of our preliminary
algorithm. Specifically, we will
prove Lemmas~\ref{lem:case1}, \ref{lem:case2}, \ref{lem:case3}, and
\ref{lem:prune}. The major analysis technique is the exchange argument,
which is quite standard for proving correctness of greedy
algorithms (e.g., see \cite{ref:KleinbergAl05Ch4}).

%For any configuration $\calC$ over a subset $\calI'\subseteq \calI$
%and any $i\in \calI'$, recall that $d(i,\calC)$ denote the displacement of $I_i$ in $\calC$.

%Consider a feasible configuration $\calC$ of all intervals of $\calI$.
Let $L$ be a list of all indices of $\calI$.
For any two indices $j,k\in [1,n]$, let $L[j,k]$ denote the sub-list
of all indices of $L$ between $j$ and $k$ (including $j$ and $k$).

For any $1\leq j<k\leq n$, we say that $(j,k)$ is an {\em inversion}
of $L$ if $x_j^r\leq x_k^r$ and $k$ is before $j$
in $L$ ($k$ and $j$ are not necessarily consecutive in $L$; e.g., see
Fig.~\ref{fig:exchange} with $L=L_{opt}$).
For an inversion $(j,k)$, we
further introduce two sets of indices $L^1[j,k]$ and $L^2[j,k]$
as follows (e.g., see Fig.~\ref{fig:exchange} with $L=L_{opt}$). Let $L^1[j,k]$ consist of all indices
$i\in L[j,k]$ such that $i<k$ and $i\neq j$; let
$L^2[j,k]$ consist of all indices
$i\in L[j,k]$ such that $i\geq k$. Hence,
$L^1[j,k]$, $L^2[j,k]$, and $\{j\}$ form a partition of the indices of
$L[j,k]$.

\begin{figure}[t]
\begin{minipage}[t]{\textwidth}
\begin{center}
\includegraphics[height=1.0in]{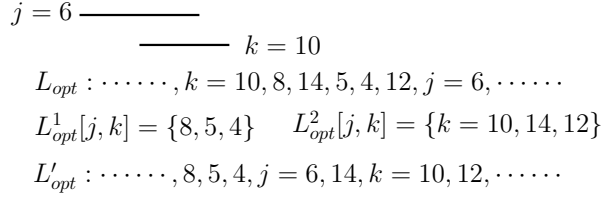}
\caption{\footnotesize Illustrating an inversion $(j,k)$ of $L_{opt}$
and an example for
Lemma~\ref{lem:inversion}: the intervals $j$ and $k$ are shown in
their input positions.}
\label{fig:exchange}
\end{center}
\end{minipage}
\vspace{-0.15in}
\end{figure}

We first give the following lemma, which will be extensively used later.

\begin{lemma}\label{lem:inversion}
Let $L_{opt}$ be an optimal list of all indices of $\calI$. If $L_{opt}$ has an
inversion $(j,k)$, then there exists another optimal list $L_{opt}'$ that is
the same as $L_{opt}$ except that the sublist $L_{opt}[j,k]$ is changed to the
following: all indices of $L_{opt}^1[j,k]$ are before $j$ and all indices of
$L_{opt}^2[j,k]$ are after $j$ (in particular, $k$ is after $j$, so $(j,k)$
is not an inversion any more in $L'_{opt}$), and further, the relative order of the
indices of $L^1_{opt}[j,k]$ in $L'_{opt}$ is the same as that in
$L_{opt}$ (but this may not be the case for $L^2_{opt}[j,k]$).
E.g., see Fig.~\ref{fig:exchange}.
%\begin{enumerate}
%\item
%
%
%$j$ is in the front of $k$ (so $(j,k)$ is not an inversion any more in
%$L'$);
%\item
%the indices of $L_1(i,j)$ are all in front of $j$ with the same
%relative order as that in $L(i,j)$;
%\item
%the indices of $L_2(i,j)$ are all after $j$ and before $k$, with the same
%relative order as that in $L(i,j)$;
%\item
%the indices of $L_3(i,j)$ are all after $k$, with the same
%relative order as that in $L(i,j)$.
%\end{enumerate}
\end{lemma}

Many proofs given later in the paper will utilize Lemma~\ref{lem:inversion}
as a basic technique for ``eliminating'' inversions
in optimal lists. Before giving the proof of Lemma~\ref{lem:inversion}, which is
somewhat technical, lengthy, and tedious, we first show that Lemma~\ref{lem:case1} can be easily
proved with the help of Lemma~\ref{lem:inversion}.

\subsection{Proof of Lemma~\ref{lem:case1}.}

Assume $L'$ is a canonical list of $\calI[1,i-1]$. Our goal is to prove
that $L$ is a canonical list of $\calI[1,i]$.

Since $L'$ is a canonical list, by the definition of a canonical list, there exists an optimal configuration
$\calC$ in which the order of the intervals of $\calI[1,i-1]$ is the same
as that in $L'$. Let $L_{opt}$ be the list of indices of the intervals
of $\calI$ in $\calC$.  If ${i}$ is after $m$ in $L_{opt}$, then
$L$  is consistent with $L_{opt}$ and thus is a canonical list of $\calI[1,i]$.
In the following, we assume $i$ is before $m$ in $L_{opt}$.

Since $m<i$, $x_m^r\leq x_i^r$, and $i$ is before $m$ in $L_{opt}$, $(m,i)$ is an
inversion in $L_{opt}$. Let $L'_{opt}$ be another optimal list
obtained by applying Lemma~\ref{lem:inversion} on $(m,i)$. Refer to Fig.~\ref{fig:lemcase1}.
We claim that $L$ is consistent
with $L'_{opt}$, which will prove that $L$ is a canonical list. We prove
the claim below.

\begin{figure}[h]
\begin{minipage}[t]{\textwidth}
\begin{center}
\includegraphics[height=0.7in]{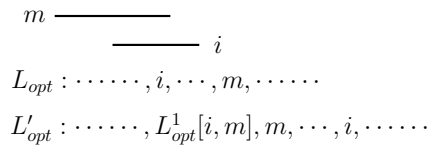}
\caption{\footnotesize Illustrating the proof of Lemma~\ref{lem:case1}.
The intervals $m$ and $i$ are shown in their input positions.}
\label{fig:lemcase1}
\end{center}
\end{minipage}
\vspace{-0.15in}
\end{figure}

Indeed, note that $L'$ is consistent with $L_{opt}$. Comparing with $L_{opt}$, by Lemma~\ref{lem:inversion}, only the
indices of the sublist $L_{opt}[m,i]$ have their relative order
changed in $L_{opt}'$. Since all indices of $L'$ are smaller than
$i$, by definition, all indices of $L'$ that are in $L_{opt}[m,i]$
are contained in $L_{opt}^1[m,i]$. By Lemma~\ref{lem:inversion}, the relative order of the indices of
$L_{opt}^1[m,i]$ in $L'_{opt}$ is the same as that in $L_{opt}$, and
further, all indices of $L_{opt}^1[m,i]$ are still before $m$ in
$L'_{opt}$. This implies that the relative order of the indices of
$L'$ does not change from $L_{opt}$ to $L'_{opt}$. Hence, $L'$ is
consistent with $L'_{opt}$. On the other hand, by Lemma~\ref{lem:inversion}, $i$
is after $m$. Thus, $L$ is consistent with $L'_{opt}$.
This proves the claim and thus proves Lemma~\ref{lem:case1}.

\subsection{Proof of Lemma~\ref{lem:inversion}}

In this section, we give the proof of Lemma~\ref{lem:inversion}.

We partition the set $L_{opt}^2[j,k]\setminus\{k\}$ into two sets $S_1$ and $S_2$,
defined as follows (e.g., see Fig.~\ref{fig:partition}).
Let $S_1$ consists of all indices $t$ of
$L_{opt}^2[j,k]\setminus\{k\}$ such that $x_t^r\leq x_j^r$ (i.e., $r_t$ is to the left of $r_j$
in the input). Let $S_2$ consists of all indices of $L_{opt}^2[j,k]\setminus\{k\}$ that are not in $S_1$.
%In fact, $S_1$ consists of all indices $t$ in
%$L_{opt}[j,k]\setminus\{j,k\}$ such that $I_t$ is contained in the
%intersection $I_j\cap I_k$ in the input;
%$S_2$ consists of all indices $t$ of
%$L_{opt}[j,k]\setminus\{j,k\}$ such that $l_t$ is strictly to the right of
%$l_k$ and $r_t$ is strictly to the right of $r_j$ in the input.
Note that $L_{opt}[j,k]=L_{opt}^1[j,k]\cup S_1\cup S_2\cup \{j,k\}$.
To simplify the notation, let $S=L_{opt}[j,k]$ and $S_0=L_{opt}^1[j,k]$ (e.g., see Fig.~\ref{fig:partition}).

\begin{figure}[t]
\begin{minipage}[t]{\textwidth}
\begin{center}
\includegraphics[height=1.2in]{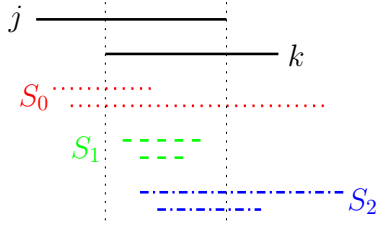}
\caption{\footnotesize Illustrating the intervals of $L_{opt}[j,k]$ in their input positions.
The two (red) dotted intervals are in $S_0=L_{opt}^1[j,k]$; the two (green) dashed intervals are in $S_1$;
the two (blue) dashed-dotted intervals are in $S_2$.}
\label{fig:partition}
\end{center}
\end{minipage}
\vspace{-0.15in}
\end{figure}

We only consider the general case where none of $S_0$, $S_1$, and
$S_2$ is empty since other cases can be analyzed by similar but
simpler techniques.

In the following, from $L_{opt}$, we will subsequently construct a sequence of
optimal lists $L_0,L_1,L_2,L_3$, such that eventually $L_3$ is the list $L'_{opt}$
specified in the statement of Lemma~\ref{lem:inversion} (e.g., see
Fig.~\ref{fig:lists}).

\subsubsection{The List $L_0$}
\label{sec:listl0}

For any adjacent indices $h$ and $g$ of $L_{opt}[j,k]\setminus\{j,k\}$
such that $h$ is before $g$ in $L_{opt}$,
we say that $(h,g)$ is an {\em exchangeable pair} if one of the three cases happen: $g\in S_0$ and $h\in S_1$; $g\in S_1$ and $h\in S_2$; $g\in S_0$ and $h\in S_2$.

In the following, we will perform certain ``exchange operations'' to
eliminate all exchangeable pairs of $L_{opt}$, after which we will
obtain another optimal list $L_0$ in which for any $i_0\in S_0$,
$i_1\in S_1$, $i_2\in S_2$, $i_0$ is before $i_1$ and $i_2$ is after
$i_1$, and all other indices of $L_0$ have the same positions as in $L_{opt}$ (e.g., see Fig.~\ref{fig:lists}).

\begin{figure}[t]
\begin{minipage}[t]{\textwidth}
\begin{center}
\includegraphics[height=0.9in]{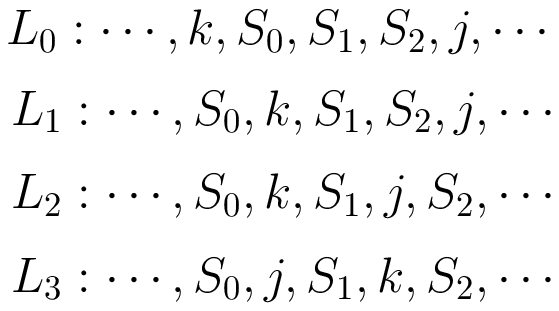}
\caption{\footnotesize Illustrating the relative order of
$k,j,S_0,S_1,S_2$ in the four lists $L_0,L_1,L_2.L_3$. }
\label{fig:lists}
\end{center}
\end{minipage}
\vspace{-0.15in}
\end{figure}

Consider any exchangeable pair $(h,g)$ of $L_{opt}$. %In the following,
Let $L'$ be another list that is the same as $L_{opt}$
except that $h$ and $g$ exchange their order.
We call this an {\em exchange operation}. In the following, we show that $L'$ is an optimal list.

Since $L_{opt}$ is an optimal list, there is an optimal configuration $\calC$ in
which the order of the intervals is the same as $L_{opt}$.
Consider the configuration $\calC'$ that is the same as $\calC$ except that we
exchange the order of $h$ and $g$ in the following way (e.g., see Fig~\ref{fig:exchangeconf}):
$x_{g}^l(\calC')=x_h^l(\calC)$ and $x_{h}^r(\calC')=x_g^r(\calC)$,
i.e., the left endpoint $l_g$ of $I_g$ in $\calC'$ is at the same position as
$l_h$ in $\calC$ and the right end point $r_h$ of $I_h$ in $\calC'$ is at the same
position as $r_g$ in $\calC$.
Clearly, the order of intervals in $\calC'$ is the same as that in $L'$.
In the following, we show that $\calC'$ is an optimal configuration, which will prove that $L'$
is an optimal list.

\begin{figure}[h]
\begin{minipage}[t]{\textwidth}
\begin{center}
\includegraphics[height=0.5in]{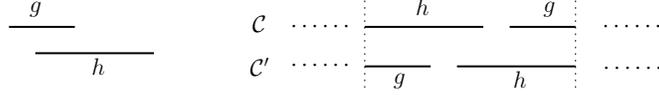}
\caption{\footnotesize Left: Illustrating the intervals $g$ and $h$ at their input positions.
Right: Illustrating the two intervals $h$ and $g$ in the
configurations $\calC$ and $\calC'$ (note that $h$ and $g$ do not have to be connected). }
\label{fig:exchangeconf}
\end{center}
\end{minipage}
\vspace{-0.15in}
\end{figure}

We first show that $\calC'$ is feasible. Recall that intervals $h$ and $g$ are adjacent
in $L_{opt}$ and also in $L'$. By our way of setting $I_g$ and $I_h$ in $\calC'$,
The segments of $\ell$ ``spanned'' by $I_h$ and $I_g$ in both
$\calC$ and $\calC'$ are exactly the same (e.g., the segments between the two vertical
dotted lines in Fig.~\ref{fig:exchangeconf}). Since no two intervals
of $\calI$ overlap in $\calC$, no two intervals overlap in $\calC'$ as well.

Next, we show that every interval of $\calI$ is valid in $\calC'$. To this end, it is
sufficient to show that $I_h$ and $I_g$ are valid in $\calC'$ since other intervals do
not change positions from $\calC$ to $\calC'$. For $I_h$, comparing with its position in $\calC$,
$I_h$ has been moved rightwards in $\calC'$, and thus $I_h$ is valid in $\calC'$.
For $I_g$, since $(h,g)$ is an exchangeable pair, $g$ is either in $S_0$ or in $S_1$. In either
case, $x_g^l\leq x_k^r$. On the other hand, $I_k$ is to the left of $I_g$ in $\calC'$, which
implies that $x_k^r(\calC')\leq x_g^l(\calC')$. Since $I_k$ does not change position from
$\calC$ to $\calC'$ and $I_k$ is valid in $\calC$, we have $x_k^r\leq x_k^r(\calC)=x_k^r(\calC')$.
Combining the above discussion, we have $x_g^l\leq x_k^r\leq x_k^r(\calC)=x_k^r(\calC')\leq x_g^l(\calC')$.
Thus, $I_g$ is valid in $\calC'$.
This proves that $\calC'$ is a feasible configuration.

We proceed to show that $\calC'$ is an optimal configuration by proving
that the max-displacement of $\calC'$ is no more than the
max-displacement of $\calC$, i.e., $\delta(\calC')\leq \delta(\calC)$.
Note that $\delta(\calC)=\delta_{opt}$ since $\calC$ is an optimal configuration.
Comparing with $\calC$, $I_g$ has been moved leftwards and $I_h$ has
been moved rightwards in $\calC'$. Therefore, to prove
$\delta(\calC')\leq \delta_{opt}$, it suffices to show that the
displacement of $I_h$ in $\calC'$, i.e., $d(h,\calC')$, is at most
$\delta_{opt}$.
Since $(h,g)$ is an exchangeable pair, $h$ is either in $S_1$ or in $S_2$. In either
case, $x_j^l\leq x_h^l$. On the other hand, $I_j$ is to the right of $I_h$ in $\calC'$, which
implies that $x_h^l(\calC')\leq x^l_j(\calC')$. Consequently, we have
$d(h,\calC')=x_h^l(\calC')-x_h^l \leq x_j^l(\calC')-x_j^l=d(j,\calC')$. Since $I_j$ does
not change position from $\calC$ to $\calC'$, $d(h,\calC')\leq d(j,\calC')
=d(j,\calC)\leq \delta_{opt}$.
This proves that $\calC'$ is an optimal configuration and $L'$ is an optimal list.

If $L'$ still has an exchangeable pair, then we keep applying the above
exchange operations until we obtain an optimal list $L_0$ that does
not have any exchangeable pairs. Hence, $L_0$ has the following
property:
for any $i_t\in S_t$ for $t=0,1,2$, $i_0$ is before $i_1$
and $i_2$ is after $i_1$, and all other indices of $L_0$ have the same
positions as in $L_{opt}$. Further, notice that our exchange operation
never changes the relative order of any two indices in $S_t$ for each
$0\leq t\leq 2$. In particular, the relative order of the indices of $S_0$
in $L_{opt}$ is the same as that in $L_0$.

\subsubsection{The List $L_1$}
\label{sec:listl1}

Let $L_1$ be another list that is the same
as $L_0$ except that $k$ is between the indices of $S_0$ and the indices of $S_1$ (e.g., see Fig.~\ref{fig:lists}).
In the following, we show that $L_1$ is also an optimal list.
This can be done by keeping performing exchange operations between
$k$ and its right neighbor in $S_0$ until all
indices of $S_0$ are to the left of $k$. The details are given below.

Let $g$ be the right neighboring index of $k$ in $L_0$ and $g$ is
in $S_0$. Let $L'$ be the list that is the same as $L_0$ except that
we exchange the order of $k$ and $g$. In the following, we show
that $L'$ is an optimal list.

Since $L_0$ is an optimal list, there is an optimal configuration
$\calC$ in which the order of the indices of the intervals is the same
as $L_0$. Consider the configuration $\calC'$ that is the same as
$\calC$ except that we exchange the order of $k$ and $g$ in the
following way: $x_g^l(\calC')=x_k^l(\calC)$ and
$x_k^r(\calC')=x_g^r(\calC)$ (e.g., see Fig.~\ref{fig:listl1}; similar to that in
Section~\ref{sec:listl0}).
In the following, we show that $\calC'$ is an optimal solution, which
will prove that $L'$ is an optimal list.

\begin{figure}[h]
\begin{minipage}[t]{\textwidth}
\begin{center}
\includegraphics[height=0.5in]{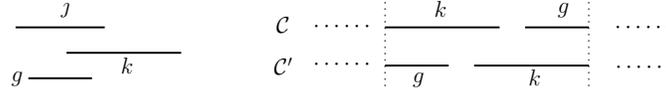}
\caption{\footnotesize Left: Illustrating the intervals $j$, $k$, and
$g$ at their input positions.
Right: Illustrating the two intervals $k$ and $g$ in
the configurations $\calC$ and $\calC'$. }
\label{fig:listl1}
\end{center}
\end{minipage}
\vspace{-0.15in}
\end{figure}

We first show that $\calC'$ is feasible. By the similar argument as
in Section~\ref{sec:listl0}, no two intervals overlap in $\calC'$. Next we show that every
interval is valid in $\calC$. It is sufficient to show that both $I_k$
and $I_g$ are valid. For $I_k$, comparing with its position in $\calC$, $I_k$ has been
moved rightwards in $\calC'$ and thus $I_k$ is valid in $\calC'$. For
$I_g$, since $g\in S_0$, by the definition of $S_0$, $x_g^l\leq x_k^l$
(e.g., see the left side of Fig.~\ref{fig:listl1}). Since
$x_k^l\leq x_k^l(\calC)= x_g^l(\calC')$, we obtain that $x_g^l\leq
x_g^l(\calC')$ and $I_g$ is valid in $\calC'$.

We proceed to show that $\calC'$ is an optimal configuration by proving
that $\delta(\calC')\leq \delta(\calC)= \delta_{opt}$. Comparing with $\calC$, $I_g$
has been moved leftwards and $I_k$ has been moved rightwards in
$\calC'$. Therefore, to prove $\delta(\calC')\leq \delta_{opt}$, it
suffices to show that $d(k,\calC')\leq \delta_{opt}$. Recall that
$l_j$ is to the left of $l_k$ in the input. Note that $k$ is to the
left of $j$ in $L'$. Hence, $l_k$ is to the left of $l_j$ in $\calC'$. Thus,
$d(k,\calC')\leq d(j,\calC')$. Note that $d(j,\calC')=d(j,\calC)$ since
the position of $I_j$ does not change from $\calC$ to $\calC'$.
Therefore, we obtain $d(k,\calC')\leq d(j,\calC)\leq \delta_{opt}$.
This proves that $\calC'$ is an optimal configuration and $L'$ is an
optimal list.

If the right neighbor of $k$ in $L'$ is still in $S_0$,
then we keep performing the above exchange until all indices of $S_0$
are to the left of $k$, at which moment we obtain the list $L_1$.
Thus, $L_1$ is an optimal list.

\subsubsection{The List $L_2$}
Let $L_2$ be another list that is the same
as $L_1$ except that $j$ is between the indices of $S_1$ and the
indices of $S_2$ (e.g., see Fig.~\ref{fig:lists}).
This can be done by keeping performing exchange
operations between $j$ and its left neighbor in $S_2$ until all
indices of $S_2$ are to the right of $j$, which is symmetric to that
in Section~\ref{sec:listl1}.
The details are given below.

Let $h$ be the left neighbor of $j$ in $L_1$ and $h$ is in $S_2$.
Let $L'$ be the list that is the same as $L_1$ except that we exchange
the order of $h$ and $j$. In the following, we show that $L'$ is an
optimal list.

Since $L_1$ is an optimal list, there is an optimal configuration
$\calC$ in which the order of the indices of the intervals is the same
as $L_1$. Consider the configuration $\calC'$ that is the same as $\calC$ except
that we exchange the order of $j$ and $h$ in the following way:
$x_j^l(\calC')=x_h^l(\calC)$ and $x_h^r(\calC')=x_j^r(\calC)$ (e.g.,
see Fig.~\ref{fig:listl2}).  In the following, we show that
$\calC'$ is an optimal solution, which will prove that $L'$ is an
optimal list.

\begin{figure}[h]
\begin{minipage}[t]{\textwidth}
\begin{center}
\includegraphics[height=0.5in]{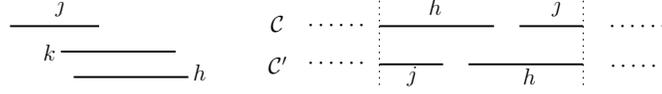}
\caption{\footnotesize Left: Illustrating the intervals $j$, $k$, and
$h$ at their input positions.
Right: Illustrating the two intervals $h$ and $j$ in
the configurations $\calC$ and $\calC'$. }
\label{fig:listl2}
\end{center}
\end{minipage}
\vspace{-0.15in}
\end{figure}

We first show that $\calC'$ is feasible. By the similar argument as
before, no two intervals overlap in $\calC'$. Next we show that every
interval is valid in $\calC'$. It is sufficient to show that both $I_j$
and $I_h$ are valid. For $I_h$, comparing with its position in $\calC$, $I_h$ has been
moved rightwards in $\calC'$ and thus $I_h$ is valid in $\calC'$. For
$I_j$, since $h\in S_2$, by the definition of $S_2$, $x_j^l\leq
x_h^l$. Since $x_h^l\leq x_h^l(\calC)= x_j^l(\calC')$, we obtain that $x_j^l\leq
x_j^l(\calC')$ and $I_j$ is valid in $\calC'$.

We proceed to show that $\calC'$ is an optimal configuration by proving
that $\delta(\calC')\leq \delta(\calC)= \delta_{opt}$. Comparing with $\calC$, $I_j$
has been moved leftwards and $I_h$ has been moved rightwards in
$\calC'$. Therefore, to prove $\delta(\calC')\leq \delta_{opt}$, it
suffices to show that $d(h,\calC')\leq \delta_{opt}$.
Since $h$ is in $S_2$, $x_j^r\leq x_h^r$. Since $x_h^r(\calC')=
x_j^r(\calC)$, we deduce $d(h,\calC')=x_h^r(\calC')-x_h^r\leq
x_j^r(\calC)-x_j^r=d(j,\calC)\leq \delta_{opt}$.
This proves that $\calC'$ is an optimal configuration and $L'$ is an
optimal list.

If the left neighbor of $j$ in $L'$ is still in $S_2$,
then we keep performing the above exchange until all indices of $S_2$
are to the right of $j$, at which moment we obtain the list $L_2$.
Thus, $L_2$ is an optimal list.

\subsubsection{The List $L_3$}
\label{sec:list3}
Let $L_3$ be the list that is the same as $L_2$ except that we exchange the order of $k$ and $j$, i.e.,
in $L_3$, the indices of $S_1$ are all after $j$ and before $k$ (e.g.,
see Fig.~\ref{fig:lists}). In the following, we prove that $L_3$ is an optimal list.

Since $L_2$ is an optimal list, there is an optimal configuration
$\calC$ in which the order of the indices of intervals is the same as
$L_2$. Consider the configuration $\calC'$ that is the same as $\calC$
except the following (e.g., see Fig.~\ref{fig:listl3}):
First, we set $x_j^l(\calC')=x_k^l(\calC)$;
%(i.e., $l_j$ is at the same position as $l_k$ in $\calC$);
second, we shift each interval of $S_1$ leftwards by distance
$|I_k|-|I_j|$ (if this value is negative, we actually shift
rightwards by its absolute value); third, we set
$x_k^r(\calC')=x_j^r(\calC)$ (i.e., $r_k$ is at the same position as
$r_j$ in $\calC$). Clearly, the interval order of $\calC'$ is the same
as $L_3$. In the following, we show that $\calC'$ is an optimal
configuration, which will prove that $L_3$ is an optimal list.

\begin{figure}[h]
\begin{minipage}[t]{\textwidth}
\begin{center}
\includegraphics[height=0.5in]{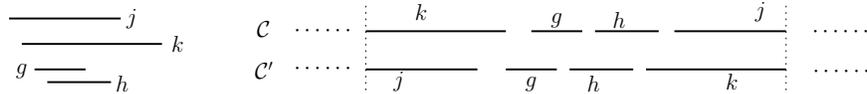}
\caption{\footnotesize Left: Illustrating the intervals $j$, $k$, $g$ and
$h$ at their input positions, where $S_1=\{g,h\}$.
Right: Illustrating the intervals of $S_1\cup \{j,k\}$ in
the configurations $\calC$ and $\calC'$. }
\label{fig:listl3}
\end{center}
\end{minipage}
\vspace{-0.15in}
\end{figure}

We first show that $\calC'$ is feasible.
By our way of setting positions of intervals in $S_1\cup \{j,k\}$,
%In this way, the segment of $\ell$ between $l_k$
%and $r_j$ in $\calC'$ is exactly that between $l_j$ and $r_k$ in $\calC$.
One can easily verify that no two intervals of $\calC'$ overlap.
Next we show that  every interval is valid in $\calC'$.
It is sufficient to show that all intervals in
$S_1\cup\{j,k\}$ are valid. Comparing with $\calC$, $I_k$ has been moved
rightwards in $\calC'$. Thus, $I_k$ is valid in $\calC'$. Recall that
$x^l_j\leq x^l_k$ and $x_j^l(\calC')=x_k^l(\calC)$. Since
$x^l_k\leq x_k^l(\calC)$ (because $I_k$ is valid in $\calC$), we
obtain that $x^l_j\leq x^l_j(\calC')$ and $I_j$ is valid in $\calC'$.
Consider any index $t\in S_1$. By the definition of $S_1$, $x_t^l\leq
x_j^r$. %(i.e., $l_t$ is to the left of $r_j$ in the input).
Since $j$
is to the left of $t$ in $\calC'$, we have $x_j^r(\calC')\leq
x_t^l(\calC')$. Since $x_j^r\leq x_j^r(\calC')$ (because $I_j$ is
valid in $\calC'$), we obtain that $x_t^l\leq x_j^r \leq
x_j^r(\calC')\leq  x_t^l(\calC')$ and thus $I_t$
is valid in $\calC'$. This proves that $\calC'$ is feasible.

We proceed to show that $\calC'$ is an optimal configuration by
proving that $\delta(\calC')\leq\delta(\calC)=\delta_{opt}$. It is sufficient to
show that for any $t\in S_1\cup\{j,k\}$, $d(t,\calC')\leq
\delta_{opt}$. Comparing with $\calC$, $I_j$ has been moved
leftwards in $\calC'$, and thus, $d(j,\calC')\leq d(j,\calC)\leq
\delta_{opt}$. Recall that $x^r_j\leq x^r_k$ and
$x_k^r(\calC')=x_j^r(\calC)$. We can deduce
$d(k,\calC')=x_k^r(\calC')-x_k^r\leq x_j^r(\calC)-x_j^r\leq
d(j,\calC)\leq \delta_{opt}$. Consider any $t\in S_1$. By the
definition of $S_1$, $x_t^l\geq x_k^l$. On the other hand, since $t$
is to the left of $k$ in $\calC'$, $x_t^l(\calC')\leq x_k^l(\calC')$.
Therefore, we obtain that $d(t,\calC')=x_t^l(\calC')-x_t^l\leq
x_k^l(\calC')-x_k^l=d(k,\calC')$. We have proved above that
$d(k,\calC')\leq \delta_{opt}$, and thus $d(t,\calC')\leq
\delta_{opt}$.
This proves that $\calC'$ is an optimal configuration and $L_3$ is an
optimal list.

Notice that $L_3$ is the list $L'_{opt}$ specified in the lemma
statement. Indeed, in
all above lists from $L_{opt}$ to $L_3$, the relative order of the
indices of $S_0$ (which is $L_{opt}^1[j,k]$) never changes.
This proves Lemma~\ref{lem:inversion}.

\subsection{Proof of Lemma~\ref{lem:case2}}

In this section, we prove Lemma~\ref{lem:case2}.
Assume $L'$ is a canonical list of $\calI[1,i-1]$. Our goal is to prove
that $L$ is also a canonical list of $\calI[1,i]$.

Since $L'$ is a canonical list, there exists an optimal configuration
$\calC$ in which the order the intervals of $\calI[1,i-1]$ is the same
as that in $L'$. Let $L_{opt}$ be the list of indices of the intervals
of $\calI$ in $\calC$.
If, in $L_{opt}$, ${i}$ is before $m$ and after every index of
$\calI[1,i-1]\setminus\{m\}$, then
$L$ is consistent with $L_{opt}$ and thus is a canonical list of
$\calI[1,i]$, so we are done with the proof.

In the following, we
assume $L$ is not consistent with $L_{opt}$. There are two cases. In
the first case, $i$ is after $m$ in $L_{opt}$. In the second case,
$i$ is before $j$ in $L_{opt}$ for some $j\in \calI[1,i-1]\setminus\{m\}$. We
analyze the two cases below. In each case, by performing certain
exchange operations and using Lemma~\ref{lem:inversion}, we will find an optimal
list of all intervals of $\calI$ such that $L$ is consistent with the
list (this will prove that $L$ is an canonical list of
$\calI[1,i]$).

\subsubsection{The First Case}
\label{sec:firstcase}
%\paragraph{The first case.}
Assume $i$ is after $m$ in $L_{opt}$. Let $S$ denote the set of
indices strictly between $m$ and $i$ in $L_{opt}$ (so neither $m$ nor
$i$ is in $S$). Since all indices of $\calI[1,i-1]$ are before $m$ in
$L_{opt}$, it holds that $j>i$ for each
index $j\in S$. Let $S'$ be the set of indices $j$ of $S$ such that
$x_j^r\geq x_{i}^r$. Note that for each $j\in S'$, the pair $(i,j)$ is an
inversion.  We consider the general case where neither $S$ nor $S'$ is
empty since the analysis for other cases is similar but easier.

Let $j$ be the rightmost index of $S'$. Again, $(i,j)$ is an
inversion. By Lemma~\ref{lem:inversion}, we can obtain another optimal
list $L'_{opt}$ such that $j$ is after $i$ and positions of the
indices other than those in $S$ are the same as before in $L_{opt}$. Further, the indices
strictly between $m$ and $i$ in $L'_{opt}$ are all in $S$. If
there is an index $j$ between $m$ and $i$ in $L'_{opt}$ such that
$(i,j)$ is an inversion, then we apply Lemma~\ref{lem:inversion}
again. We do this until we obtain an optimal list $L_0$ in which for
any index $j$ strictly between $m$ and $i$, $(i,j)$ is  not an
inversion, and thus $x_j^r<x_{i}^r$ (this further implies that $I_j$ is contained in $I_{i}$ in the input as $i<j$).
Let $S_0$ denote the set of indices strictly between $m$ and $i$ in $L_0$.

Consider the list $L_1$ that is the same as $L_0$ except that we
exchange the positions of $m$ and $i$, i.e., the indices of $S_0$
are now after $i$ and before $m$. In the following, we prove that
$L_1$ is an optimal list. Note that $L$ is consistent with $L_1$,
and thus once we prove that $L_1$ is an optimal list, we also prove
that $L$ is a canonical list of $\calI[1,i]$.
The technique for proving that $L_1$ is an optimal list is
similar to that in Section~\ref{sec:list3}. The details are given below.

Since $L_0$ is an optimal list, there is an optimal configuration
$\calC$ in which the order of the indices of intervals is the same as
$L_0$. Consider the configuration $\calC'$ that is the same as $\calC$
except the following (e.g., see Fig.~\ref{fig:firstcase}):
First, we set $x_{i}^l(\calC')=x_m^l(\calC)$;
%(i.e., $l_{i}$ is at the same position as $l_m$ in $\calC$);
second, we shift each interval of $S_0$ leftwards by distance
$|I_m|-|I_{i}|$ (again, if this value is negative, we actually shift
rightwards by its absolute value); third, we set $x_m^r(\calC')=x_{i}^r(\calC)$.
%(i.e., $r_m$ is at the same position as $r_{i}$ in $\calC$).
Clearly, the interval order in $\calC'$ is the same as $L_1$.
In the following, we show that $\calC'$ is an optimal
configuration, which will prove that $L_1$ is an optimal list.

\begin{figure}[h]
\begin{minipage}[t]{\textwidth}
\begin{center}
\includegraphics[height=0.5in]{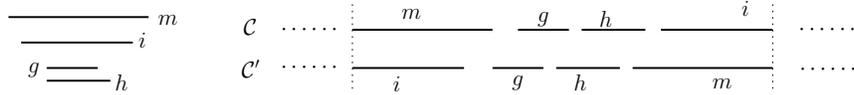}
\caption{\footnotesize Left: Illustrating the intervals $j$, $k$, $g$ and
$h$ at their input positions, where $S_0=\{g,h\}$.
Right: Illustrating the intervals of $S_0\cup \{m,i\}$ in
the configurations $\calC$ and $\calC'$. }
\label{fig:firstcase}
\end{center}
\end{minipage}
\vspace{-0.15in}
\end{figure}

We first show that $\calC'$ is feasible.
%In this way, the segment of $\ell$ between $l_m$ and $r_{i}$ in $\calC'$ is exactly that
%between $l_{i}$ and $r_m$ in $\calC$.
As in Section~\ref{sec:list3},
no two intervals of $\calC'$ overlap. Next, we show that every interval
is valid in $\calC'$. It is sufficient to show that all intervals in
$S_0\cup\{m,i\}$ are valid since other intervals do no change
positions from $\calC$ to $\calC'$. Comparing with its position in
$\calC$, $I_m$ has been moved
rightwards in $\calC'$. Thus, $I_m$ is valid in $\calC'$.
Recall that in Case II of our algorithm, it holds that $x_{i}^l\leq
x_m^l(\calC_{L'})$, where $\calC_{L'}$ is the configuration of only the
intervals of $\calI[1,i-1]$ following their order in $L'$.
Since $\calC_{L'}$ is the configuration
constructed by the left-possible placement strategy and
the order of the indices of $\calI[1,i-1]$ in $\calC$ is the same as $L'$,
it holds that $x_m^l(\calC_{L'})\leq x_m^l(\calC)$. Hence, we obtain
$x_{i}^l\leq x_m^l(\calC)$. Since $x_{i}^l(\calC')=x_m^l(\calC)$,
$x_{i}^l\leq x_{i}^l(\calC')$ and $I_{i}$ is valid in
$\calC'$.
Consider any index $j\in S_0$. Recall that $I_j$ is contained in
$I_{i}$ in the input. Thus, $x_j^l\leq x_{i}^r$. Since $i$
is to the left of $j$ in $\calC'$, we have $x_{i}^r(\calC')\leq
x_j^l(\calC')$. Since $x_{i}^r\leq x_{i}^r(\calC')$ (because $I_{i}$ is
valid in $\calC'$), we obtain that $x_j^l\leq x_j^l(\calC')$ and $I_j$
is valid in $\calC'$.
This proves that $\calC'$ is feasible.

We proceed to show that $\calC'$ is an optimal configuration by
proving that $\delta(\calC')\leq \delta(\calC)=\delta_{opt}$. It
suffices to show that for any $j\in S_0\cup\{m,i\}$, $d(j,\calC')\leq
\delta_{opt}$. Comparing with $\calC$, $I_{i}$ has been moved
leftwards in $\calC'$, and thus $d(i,\calC')\leq d(i,\calC)\leq
\delta_{opt}$. Since $x^r_{i}\leq x^r_m$ and
$x_m^r(\calC')=x_{i}^r(\calC)$, we can deduce
$d(m,\calC')=x_m^r(\calC')-x_m^r\leq x_{i}^r(\calC)-x_{i}^r=
d(i,\calC)\leq \delta_{opt}$. Consider any $j\in S_0$. Recall that
$x_j^l\geq x_{i}^l\geq  x_m^l$. On the other hand, since $j$
is to the left of $m$ in $\calC'$, $x_j^l(\calC')\leq x_m^l(\calC')$.
Therefore, $d(j,\calC')=x_j^l(\calC')-x_j^l\leq
x_m^l(\calC')-x_m^l=d(m,\calC')$. We have proved above that
$d(m,\calC')\leq \delta_{opt}$, and thus $d(j,\calC')\leq
\delta_{opt}$.

This proves that $\calC'$ is an optimal configuration and $L_1$ is an
optimal list. As discussed above, this also proves that $L$ is a
canonical list of $\calI[1,i]$. This finishes the proof of the lemma
in the first case.

\subsubsection{The Second Case}
\label{sec:secondcase}

In the second case, $i$ is before $j$ in $L_{opt}$ for some $j\in
\calI[1,i-1]\setminus \{m\}$. We assume there is no other indices of
$\calI[1,i-1]$ strictly between $i$ and $j$ in $L_{opt}$ (otherwise, we take
$j$ as the leftmost such index to the right of $i$).

Let $\widehat{L_0}$ be the list of indices of $\calI[1,i]$ following their
order in $L_{opt}$. Therefore, $\widehat{L_0}$ is a canonical list.
Let $\widehat{L_1}$ be the list the same as $\widehat{L_0}$ except that the order of $i$ and
$j$ is exchanged. In the following, we first show that $\widehat{L_1}$
is also a canonical list of $\calI[1,i]$.  The
proof technique is very similar to the above first case.

Let $S$ denote the set of indices strictly between $i$ and $j$ in $L_{opt}$. By the
definition of $j$, $k>i>j$ holds for each
index $k\in S$. Let $S'$ be the set of indices $k$ of $S$ such that
$x_k^r\geq x_{j}^r$. Hence, for each $k\in S'$, the pair $(j,k)$ is an
inversion of $L_{opt}$. We consider the general case where neither $S$ nor $S'$ is
empty (otherwise the proof is similar but easier).

As in Section~\ref{sec:firstcase}, starting from the rightmost index of $S'$,
we keep applying Lemma~\ref{lem:inversion} to the inversion pairs
and eventually obtain an optimal list $L_0$ in which for any
index $k$ of $L_0$ strictly between $i$ and $j$, $(j,k)$ is not an
inversion and thus $x_k^r<x_j^r$ (hence $I_k\subseteq I_j$ in the
input as $j<k$). Let $S_0$ denote the set of indices strictly
between $i$ and $j$ in $L_0$.

Consider the list $L_1$ that is the same as $L_0$ except that we
exchange the positions of $i$ and $j$, i.e., the indices of $S_0$
are now after $j$ and before $i$. In the following, we prove that
$L_1$ is an optimal list, which will also prove that $\widehat{L_1}$ is
a canonical list of $\calI[1,i]$ since $\widehat{L_1}$ is consistent with
$L_1$.

Since $L_0$ is an optimal list, there is an optimal configuration
$\calC$ in which the order of the intervals is the same as
$L_0$. Consider the configuration $\calC'$ that is the same as $\calC$
except the following (e.g., see Fig.~\ref{fig:secondcase}):
First, we set $x_{j}^l(\calC')=x_{i}^l(\calC)$;
second, we shift each interval of $S_0$ leftwards by distance
$|I_{i}|-|I_{j}|$; third, we set $x_{i}^r(\calC')=x_{j}^r(\calC)$.
Clearly, the interval order of $\calC'$ is the same as $L_1$.
Below, we show that $\calC'$ is an optimal configuration, which will
prove that $L_1$ is an optimal list.

\begin{figure}[h]
\begin{minipage}[t]{\textwidth}
\begin{center}
\includegraphics[height=0.6in]{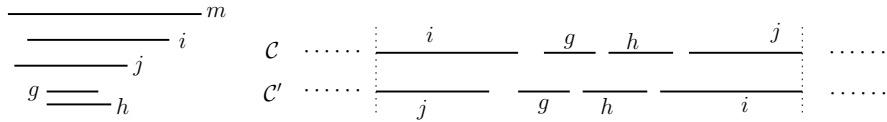}
\caption{\footnotesize Left: Illustrating five intervals
at their input positions, where $S_0=\{g,h\}$.
Right: Illustrating the intervals of $S_0\cup \{i,j\}$ in
the configurations $\calC$ and $\calC'$. }
\label{fig:secondcase}
\end{center}
\end{minipage}
\vspace{-0.15in}
\end{figure}

We first show that $\calC'$ is feasible. As before, no two intervals of $\calC'$ overlap.
Next we prove that all intervals in $S_0\cup\{i,j\}$ are valid in
$\calC'$. Comparing with its position in $\calC$, $I_{i}$ has been moved
rightwards in $\calC'$ and thus is valid. %Thus, $I_{i}$ is valid in $\calC'$.
Since $j<i$, $x_j^l<x_{i}^l$.
Since $x_{j}^l(\calC')=x_{i}^l(\calC)$ and $x_{i}^l\leq
x_{i}^l(\calC)$ (because $I_{i}$ is valid in $\calC$),
we obtain $x_{j}^l\leq x_{j}^l(\calC')$ and $I_{j}$ is valid in
$\calC'$. Consider any index $k\in S_0$.
Recall that $x_k^l\leq x_k^r \leq x_{j}^r$. Since $k$
is to the right of $j$ in $\calC'$, we have $x_{j}^r(\calC')\leq
x_k^l(\calC')$. Since $x_{j}^r\leq x_{j}^r(\calC')$,
we obtain that $x_k^l\leq x_k^l(\calC')$ and $I_k$
is valid in $\calC'$.
This proves that $\calC'$ is feasible.

We proceed to show that $\calC'$ is an optimal configuration by
proving that for any $k\in S_0\cup\{i,j\}$, $d(k,\calC')\leq
\delta(\calC) = \delta_{opt}$. Comparing with $\calC$, $I_{j}$ has been moved
leftwards in $\calC'$, and thus $d(j,\calC')\leq d(j,\calC)\leq
\delta_{opt}$. Since $m<i$,
$l_{m}$ is to the left of $r_{i}$ in the input. Since $I_m$ is to the
right of $I_{i}$ in $\calC'$, $l_m$ is to the right of $r_{i}$ in
$\calC'$.
This implies that $d(i,\calC')\leq d(m,\calC')$. Since $I_m$ does
not change position from $\calC$ to $\calC'$,
$d(m,\calC')=d(m,\calC)\leq \delta_{opt}$. Thus, we obtain
$d(i,\calC')\leq \delta_{opt}$. Consider any $k\in S_0$.
Since $i<k$, $x^l_{i}\leq x^l_k$.  On the other hand, since $k$
is to the left of $i$ in $\calC'$, $x_k^l(\calC')\leq x_{i}^l(\calC')$.
Therefore, we deduce $d(k,\calC')=x_k^l(\calC')-x_k^l\leq
x_{i}^l(\calC')-x_{i}^l=d(i,\calC')$. We have proved above that
$d(i,\calC')\leq \delta_{opt}$, and thus $d(k,\calC')\leq
\delta_{opt}$.

This proves that $\calC'$ is an optimal configuration and $L_1$ is an
optimal list. As discussed above, this also proves that
$\widehat{L_1}$ is a canonical list of $\calI[1,i]$.

If the right neighbor $j$ of $i$ in $\widehat{L_1}$ is not $m$, then
by the same analysis as above, we can show that the list obtained by
exchanging the order of $i$ and $j$ is still a canonical list of
$\calI[1,i]$. We keep applying the above
exchange operation until we obtain a canonical list $\widehat{L_2}$ of
$\calI[1,i]$ such that the right
neighbor of $i$ in $\widehat{L_2}$ is $m$. Note that $\widehat{L_2}$ is exactly
$L$, and thus this proves that $L$ is a canonical list of $\calI[1,i]$.
This finishes the proof for the lemma in the second case.

Lemma~\ref{lem:case2} is thus proved.

\subsection{Proof of Lemma~\ref{lem:case3}}

We prove Lemma~\ref{lem:case3}.
Assume that $L'$ is a canonical list of $\calI[1,i-1]$. Our goal is to prove
that either $L$ or $L^*$ is a canonical list of $\calI[1,i]$.

As $L'$ is a canonical list, there exists an optimal list $L_{opt}$ of
$\calI$ whose interval order is consistent with $L'$.
Let $\widehat{L_0}$ be the list of indices of
$\calI[1,i]$ following the same order in $L_{opt}$. If $\widehat{L_0}$ is
either $L$ or $L^*$, then we are done with the proof. Otherwise, $i$
must be before $j$ in $\widehat{L_0}$ for some index $j\in
\calI[1,i-1]\setminus\{m\}$. By using the same proof as in
Section~\ref{sec:secondcase}, we can show that
$L^*$ is a canonical list of $\calI[1,i]$. We omit the details.

%Note that $L_1$ is a canonical list of $\calI[1,i+1]$.
%We can show that the list $L_1$ obtained by exchanging the order of
%$i+1$ and $j$ in $L_0$ is also a canonical list of $\calI[1,i+1]$.

\subsection{Proof of Lemma~\ref{lem:prune}}

In this section, we prove Lemma~\ref{lem:prune}.
Assume $\calL^*$ has a canonical list $L_0$ of $\calI[1,i]$.
Recall that $L^*_{min}$ is the list of $\calL^*$ with the smallest max-displacement.
Our goal is to
prove that $L^*_{min}$ is also a canonical list of $\calI[1,i]$.

Recall that for each list
$L\in \calL^*$, $i$ and $m$ are the last two indices with $m$ at the
end, and further, in the configuration $\calC_L$ (which is obtained by the
left-possible placement strategy on the intervals in $\calI[1,i]$
following their order in $L$), $x^l_{i}(\calC_L)=x_i^l$ and
$x^l_m(\calC_L)=x_i^r$. Also, each list of $\calL^*$ is generated in
Case III of the algorithm and we have $I_{i}\subseteq I_m$ in the input.

Since $L_0$ is a canonical list of $\calI[1,i]$, there is an optimal list
$L_{opt}$ of $\calI$ that is consistent with $L_0$.
Let $S$ be the set of indices of $\calI[i+1,n]$ before $i$ in $L_{opt}$.
We consider the general case where $S$ is not empty (otherwise
the proof is similar but easier). Let $j$ be the rightmost index
of $S$ in $L_{opt}$. Let $L'_{opt}$ be the list that is the same as
$L_{opt}$ except that we move $j$ right after $i$.
In the following, we show that $L'_{opt}$ is also an optimal list.

Since $L_{opt}$ is an optimal list, there is an optimal configuration
$\calC$ in which the order of the indices of intervals is the same as
$L_{opt}$. Recall that $L_{opt}[j,i]$ is consists of indices of $L_{opt}$ between $j$ and $i$
inclusively.
Consider the configuration $\calC'$ that is the same as
$\calC$ except the following (e.g., see Fig.~\ref{fig:prune}):
First, for each index $k\in L_{opt}[j,i]\setminus\{j\}$,
move $I_k$ leftwards by distance $|I_j|$; second, move $I_j$ rightwards such
that $l_j$ is at $r_{i}$ (after $I_{i}$ is moved leftwards in the
above first step, so that $I_i$ is connected with $I_j$).
Note that the order of intervals of $\calI$ in
$\calC'$ is exactly $L'_{opt}$. In the following, we show that
$\calC'$ is an optimal configuration, which will also prove that $L'_{opt}$ is an
optimal list.

\begin{figure}[h]
\begin{minipage}[t]{\textwidth}
\begin{center}
\includegraphics[height=0.5in]{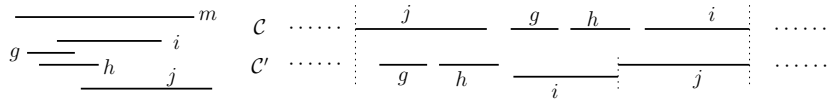}
\caption{\footnotesize Left: Illustrating five intervals
at their input positions, where $L_{opt}[j,i]=\{j,g,h,i\}$.
Right: Illustrating the intervals of $L_{opt}[j,i]$ in
the configurations $\calC$ and $\calC'$. (Interval $i$ is shifted downwards in order to
visually separate it from interval $j$.)}
\label{fig:prune}
\end{center}
\end{minipage}
\vspace{-0.15in}
\end{figure}

We first show that $\calC'$ is feasible.
By our way of setting the positions of intervals in $L_{opt}[j,i]$, no two intervals overlap
in $\calC'$. Next, we show that every interval
is valid in $\calC'$. It is sufficient to show that $I_k$ is valid in $\calC'$ for
every index $k$ in $L_{opt}[j,i]$ since all other intervals do not
move from $\calC$ to $\calC'$. Comparing with its position in $\calC$, $I_{j}$ has
been moved rightwards in $\calC'$ and thus is valid. Suppose $k\neq j$.
By the definition of $j$, $k<j$ and thus
$x_k^l\leq x_j^l$. By our way of constructing $\calC'$,
$x_j^l(\calC)\leq x_k^l(\calC')$. Since $I_j$ is valid in $\calC$, it
holds that $x_j^l\leq x_j^l(\calC)$. Thus, we obtain that $x_k^l\leq
x_k^l(\calC')$ and $I_k$ is valid. This proves that $\calC'$ is
feasible.

We proceed to show that $\calC'$ is an optimal configuration by
proving that $\delta(\calC')\leq \delta(\calC)=\delta_{opt}$. It is sufficient to
show that for any index $k\in L_{opt}[j,i]$, $d(k,\calC')\leq
\delta_{opt}$. If $k$ is not $j$, then comparing with $\calC$,
$I_{k}$ has been moved leftwards, and thus $d(k,\calC')\leq
d(k,\calC)\leq \delta_{opt}$. In the following, we show that
$d(j,\calC')\leq \delta_{opt}$. Indeed, since $m<i<j$,
it holds that $x_m^l\leq x_j^l$. On the other hand, $I_m$ is to
the right of $I_j$ in $\calC'$, and thus, $x_j^l(\calC')\leq x_m^l(\calC')$.
Therefore, we have $d(j,\calC')=x_j^l(\calC')-x_j^l\leq
x_m^l(\calC')-x_m^l=d(m,\calC')$. Since the position of $I_m$ is the
same in $\calC$ and $\calC'$, $d(m,\calC')=d(m,\calC)\leq
\delta_{opt}$. Thus, we have $d(j,\calC')\leq \delta_{opt}$.
This proves that $\calC'$ is an optimal configuration and $L_{opt}'$
is an optimal list.

If there are still indices of $\calI[i+1,n]$ before $i$ in
$L_{opt}'$, then we keep applying the above exchange operations until
we obtain an optimal list $L_{opt}''$ that does not have any index of
$\calI[i+1,n]$ before $i$, and in other words, the indices of
$L_{opt}''$ before $i$ are exactly those in
$\calI[1,i-1]\setminus\{m\}$.

Since $L''_{opt}$ is an optimal list, there is an optimal
configuration $\calC''$ whose interval order is the same as $L''_{opt}$. Let
$\calC'''$ be a configuration that is the same as $\calC''$  except the
following: For each interval $I_k$ with $k\in
\calI[1,i-1]\setminus\{m\}$, we set its position the same as its
position in $\calC_{L_{min}^*}$ (which is the configuration obtained by our
algorithm for the list $L_{min}^*$). Recall that the position
of $I_{i}$ in $\calC_{L_{min}^*}$ is the same as that in the input. On the
other hand, $x_{i}^l\leq x_{i}^l(\calC'')$. Therefore, $\calC'''$
is still a feasible configuration. We claim that $\calC'''$ is also an
optimal configuration. To see this, the maximum displacement of all
intervals in $\calI[1,i-1]\setminus\{m\}$ in $\calC'''$ is at most
$\delta(\calC_{L_{min}^*})$. Recall that $\delta(\calC_{L_{min}^*})\leq
\delta(\calC_{L_0})$. Further, since $L_0$ is a canonical list, it
holds that $\delta(\calC_{L_0})\leq \delta_{opt}$. Thus, we obtain
$\delta(\calC_{L_{min}^*})\leq \delta_{opt}$. Consequently,
the maximum displacement of all
intervals in $\calI[1,i-1]\setminus\{m\}$ in $\calC'''$ is at most $\delta_{opt}$.
Since only intervals of $\calI[1,i-1]\setminus\{m\}$ in $\calC'''$ change positions from $\calC''$
to $\calC'''$, we obtain $\delta(\calC''')\leq \delta_{opt}$ and thus $\calC'''$ is an optimal
configuration.

According to our construction of $\calC'''$, the order of the intervals
of $\calI[1,i]$ in $\calC'''$ is exactly $L_{min}^*$. Therefore, $L_{min}^*$ is
a canonical list of $\calI[1,i]$. This proves Lemma~\ref{lem:prune}.

\section{The Improved Algorithm}
\label{sec:improve}

In this section, we improve our preliminary algorithm to
$O(n\log n)$ time and $O(n)$ space. The key idea is that based on new
observations we are able to prune some ``redundant''
lists from $\calL$ after each step of the algorithm (actually
Lemma~\ref{lem:prune} already gives an example for pruning redundant
lists). More importantly,
although the number of remaining lists in $\calL$ can
still be $\Omega(n)$ in the worst case, the remaining lists of $\calL$
have certain monotonicity properties such that we are able to
implicitly maintain them in $O(n)$ space and update them in $O(\log
n)$ amortized time for each step of the algorithm for processing an interval $I_i$.

In the following, we first give some observations that will help us
to perform the pruning procedure on $\calL$.

\subsection{Observations}
In this section, unless otherwise stated, let $\calL$ be the set
after a step of our preliminary algorithm for
processing an interval $i$.
Recall that for each list $L\in \calL$, we also have a configuration
$\calC_L$ that is built following the left-possible placement
strategy. We
%use $\alpha_L$ to denote the last index of $L$ and
use $x(\calC_L)$ to denote the $x$-coordinate of the right endpoint of the rightmost
interval of $L$ in $\calC_L$.

%Roughly speaking, our observations and our pruning algorithm given later
%together imply the following: For any two lists $L_1$ and $L_2$ in $\calL$,
%if $x(\calC_{L_1})\leq x(\calC_{L_2})$
%and $\delta(\calC_{L_1})\leq \delta(\calC_{L_2})$, then $L_2$ is redundant and can be pruned.
%It should be pointed out that this property only applies to the lists
%in $\calL$ and is not applicable to any two lists of indices in
%general.

%Recall that $\calL$ refers to the set of lists maintained by our
%preliminary algorithm after interval $i+1$ is processed. It is not
%difficult to see that $i+1$ is always one of the last two indices of
%every list of $\calL$.
For any two lists $L_1$ and $L_2$ of $\calL$,
we say that $L_1$ {\em dominates} $L_2$ if the following holds: If
$L_2$ is a canonical list of $\calI[1,i]$, then $L_1$ must also
be a canonical list of $\calI[1,i]$. Hence,
if $L_1$ dominates $L_2$, then $L_2$ is ``redundant'' and can be pruned
from $\calL$.

The subsequent two lemmas give ways to identify redundant lists
from $\calL$. In general, Lemma~\ref{lem:prune1} is for the case where
two lists have different last indices while Lemma~\ref{lem:prune2} is
for the case where two lists have the same last index (notice the
slight differences in the lemma conditions).

\begin{lemma}\label{lem:prune1}
Suppose $L_1$ and $L_2$ are two lists of $\calL$ such that the last
index of $L_1$ is $m'$, the last index of $L_2$ is $m$ (with $m\neq m'$), and $x_{m'}^r\leq x_m^r$.
Then, if $\delta(\calC_{L_1})\leq d(m,\calC_{L_2})$ and
$x(\calC_{L_1})\leq x(\calC_{L_2})$, then $L_1$ dominates $L_2$.
\end{lemma}
\begin{proof}
Assume $L_2$ is a canonical list of $\calI[1,i]$. Our goal is to prove that $L_1$ is also
a canonical list of $\calI[1,i]$. It is sufficient to
construct an optimal configuration in which the order the intervals of
$\calI[1,i]$ is $L_1$. We let $h$ denote the left neighboring index
of $m'$ in $L_1$ and let $g$ denote the left neighboring index of $m$
in $L_2$.

Since $L_2$ is a canonical list, there is an optimal list $Q$
that is consistent with $L_2$.  Let $S$ denote the set of indices of $\calI[i+1,n]$
before $g$ in $Q$. We consider the general case where $S$ is not
empty (otherwise the proof is similar but easier).

By the similar analysis as in the proof of Lemma~\ref{lem:prune} (we omit the details), we
can obtain an optimal list $Q_1$ that is the same as $Q$
except that all indices of $S$ are now right after $g$ in $Q_1$
(i.e., all indices of $Q$ before $g$ except those in $S$ are
still before $g$ in $Q_1$ with the same relative order, and
all indices of $Q$ after $g$ are now after indices of $S$ in
$Q_1$ with the same relative order).
Therefore, in $Q_1$, the indices before $g$ are exactly those in $\calI[1,i]\setminus\{m\}$.

Recall that $Q_1[g,m]$ denote the sublist of $Q_1$ between $g$ and $m$
including $g$ and $m$.
%Note that for each $j\in Q_1[g,m]\setminus\{g,m\}$, $i<j$.
If there is an index $j$ in $Q_1[g,m]$ such that $(m,j)$ is an
inversion, then as in the proof of Lemma~\ref{lem:case1},
we keep applying Lemma~\ref{lem:inversion} on all such indices $j$
from right to left to obtain another optimal list $Q_2$ such
that for each $j\in Q_2[g,m]$, $(m,j)$ is not an inversion.
%(this implies that $x_j^r<x_m^r$ since $m\leq i<j$).
Note that the indices before and including $g$ in $Q_1$ are the same as those in $Q_2$.
%Also, the indices of $Q_2$
%before $g$ are excatly those in $L_1$ before $m$ (possibly with different order).
Let $S'$ denote the set of indices of $Q_2[g,m]\setminus\{g,m\}$.
Again, we consider the general case where $S'$ is not empty.
Note that $S'\subseteq \calI[i+1,n]$.
For each $j\in S'$, since $(m,j)$ is not an inversion and $m<j$, it
holds that $x_j^r<x_m^r$.

Let $Q_3$ be another list that is the same as $Q_2$ except
the following (e.g., see Fig~\ref{fig:prune1}):
First, we move $m'$ right after the indices of $S'$ and
move $m$ before the indices of $S'$ (i.e., the indices of $Q_3$ from
the beginning to $m'$ are indices of $\calI[1,i]\setminus\{m'\}$,
indices of $S'$, and $m'$);
%the relative order of
%which $m'$ is after the indices of $S'$ and $m$ is before $g$ (also,
%the indices before and including $g$ are the same as the indices of
%$L_1$ before the last index $i$, but with different orders);
second, we re-arrange the indices of $\calI[1,i]\setminus\{m'\}$
(which are all before indices of $S'$ in $Q_3$) in exactly the same order as in $L_1$.
%Note that the first $i$
%indices of $Q_3$ are exactly the same as the first $i$ indices
%of $L_1$ with the same order ($i$ is the last index of $L_1$).
In this way, $L_1$  is consistent with $Q_3$.
In the following, we show that $Q_3$  is an optimal list, which
will prove that $L_1$ is a canonical list of $\calI[1,i]$ and thus
prove the lemma.

\begin{figure}[h]
\begin{minipage}[t]{\textwidth}
\begin{center}
\includegraphics[height=0.4in]{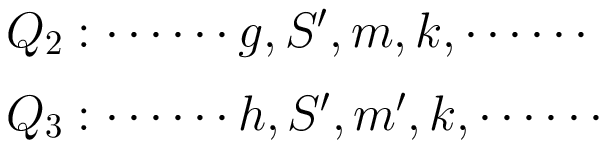}
\caption{\footnotesize Illustrating the two lists $Q_2$ and $Q_3$,
where $k$ is the right neighboring index of $m$ in $Q_2$ and $k$ is
also right neighboring index of $m'$ in $Q_3$. In $Q_2$ (resp., $Q_3$), the indices
strictly before $S'$ are exactly those in
$\calI[1,i]\setminus\{m\}$ (resp., $\calI[1,i]\setminus\{m'\}$).}
\label{fig:prune1}
\end{center}
\end{minipage}
\vspace{-0.15in}
\end{figure}

Since $Q_2$ is an optimal list, there is an optimal
configuration $\calC_2$ whose interval order is $Q_2$. Consider
the configuration $\calC_3$ whose interval order follows
$Q_3$ and whose interval positions are the same as those in $\calC_2$
except the following: First, for each index $j\in
\calI[1,i]\setminus\{m'\}$, we set the position of $I_j$ in the same
as its position in $\calC_{L_1}$ (i.e., the configuration obtained by
our algorithm for $L_1$); second, we place the intervals of $S'$ such that they
do not overlap but connect together (i.e., the right endpoint
co-locates with the left endpoint of the next interval) following their order in
$Q_2$ and the left endpoint of the leftmost
interval of $S'$ is at the right endpoint of $I_h$ (recall that $h$ is
the left neighbor of $m'$ in $L_1$, which is also the rightmost
interval of $\calI[1,i]\setminus\{m'\}$ in $Q_3$;  e.g., see
Fig.~\ref{fig:prune1}); third, we set the
left endpoint of $I_{m'}$ at the right endpoint of the rightmost
interval of $S'$. Therefore, all intervals before and including $m'$
do not have any overlap in $\calC_3$, and the intervals of $S'\cup\{h,m'\}$
essentially connect together.  In the
following, we show that $\calC_3$ is an optimal configuration, which
will prove that $Q_3$ is an optimal list.

We first show that $\calC_3$ is feasible.
We begin with proving that no two intervals overlap. Let $k$ be the
right neighboring interval of $m$ in $Q_2$ (e.g., see
Fig.~\ref{fig:prune1}), and $k$ now becomes
the right neighboring interval of $m'$ in $Q_3$.  To prove no
two intervals of $\calC_3$ overlap, it is sufficient to show that
$I_{m'}$ and $I_k$  do not overlap, i.e., $x_{m'}^r(\calC_3)\leq x_k^l(\calC_3)$.
Note that $x_k^l(\calC_3) =  x_k^l(\calC_2)$ and
$x_m^r(\calC_2)\leq x_k^l(\calC_2)$. Hence, it suffices to prove
$x_{m'}^r(\calC_3)\leq x_m^r(\calC_2)$.

We claim that in the configuration $\calC_{L_1}$, $l_{m'}$ is at $r_{h}$. Indeed, since
$x_{m'}^r\leq x_m^r$ and $I_m$ is to the left of $I_{m'}$ in
$\calC_{L_1}$, it holds that $x^l_{m'}\leq x^l_{m'}(\calC_{L_1})$.
Since $\calC_{L_1}$ is constructed based on the
left-possible placement strategy, we have
$x_{m'}^l(\calC_{L_1})=x_h^r(\calC_{L_1})$, which proves the claim.

Recall that by the definition of $x(\calC_{L_1})$, we have
$x(\calC_{L_1})=x_{m'}^r(\calC_{L_1})$.

Let $l$ be the total length of all intervals of $S'$. By our
way of constructing $\calC_3$, it holds that
$x_{m'}^r(\calC_3)=x^r_{m'}(\calC_{L_1})+l=x(\calC_{L_1})+l$.
On the other hand, since
$L_2$ is consistent with $Q_2$ and $\calC_{L_2}$ is constructed
based on the left-possible placement strategy, it holds that
$x(\calC_{L_2}) + l \leq x_m^r(\calC_2)$. By the lemma condition,
$x(\calC_{L_1})\leq x(\calC_{L_2})$. Hence, we obtain
$x_{m'}^r(\calC_3)= x(\calC_{L_1})+l\leq x(\calC_{L_2}) + l
\leq x_m^r(\calC_2)$. Thus, $I_{m'}$ and $I_k$ do not overlap in $\calC_3$.
%This proves that no two intervals overlap in $\calC_3$.

We proceed to prove that every interval of $\calC_3$ is valid.
For any interval before $h$ and including $h$ in $Q_3$, since its position in
$\calC_3$ is the same as that in $\calC_{L_1}$, it is valid. For interval
$m'$, since it is valid in $\calC_{L_1}$ and $x_{m'}^r(\calC_3)=x_{m'}^r(\calC_{L_1})+l$,
it is also valid in $\calC_3$. Consider any interval $j\in S'$.
Recall that $x_j^r< x_m^r$. Since $I_m$ is
to the left of $I_j$ in $\calC_3$, comparing with its input position,
$I_j$ must have been moved rightwards
in $\calC_3$. Thus, $I_j$ is valid. For any interval
after $m'$, its position is the same as in $\calC_2$, and thus it is
valid.

The above proves that $\calC_3$ is feasible. In the following, we show
that $\calC_3$ is an optimal configuration by proving that
$\delta(\calC_3)\leq \delta(\calC_2) = \delta_{opt}$. It is sufficient
to show that for any interval $j$ before and including $m'$ in
$\calC_3$, $d(j,\calC_3)\leq \delta_{opt}$.

\begin{itemize}
\item
Consider any interval $j$ before and including $h$ in $\calC_3$. We have
$d(j,\calC_3)=d(j,\calC_{L_1})\leq \delta(\calC_{L_1})$. By lemma
condition, $\delta(\calC_{L_1}) \leq d(m,\calC_{L_2})\leq \delta(\calC_{L_2})$.
Since $L_2$ is consistent with $Q_2$ and $\calC_{L_2}$ is constructed
based on the left-possible placement strategy, it holds that
$\delta(\calC_{L_2})\leq \delta_{opt}$. Therefore, $d(j,\calC_3)\leq
\delta_{opt}$.

\item
Consider interval $m'$. In the following, we show that
$d(m',\calC_3)\leq d(m,\calC_2)$, which will lead to
$d(m',\calC_3)\leq \delta_{opt}$ since $d(m,\calC_2)\leq
\delta_{opt}$.

By lemma condition, $d(m',\calC_{L_1})\leq \delta(\calC_{L_1})\leq
d(m,\calC_{L_2})$. As discussed above, %According to our way of contructing $\calC_3$,
$x_{m'}^r(\calC_3)=x_{m'}^r(\calC_{L_1})+l$. Therefore,
$d(m',\calC_3)=d(m',\calC_{L_1})+l$.
On the other hand, as discussed above, $x_{m}^r(\calC_2)\geq
x_{m}^r(\calC_{L_2})+l$. Therefore, $d(m,\calC_2)\geq d(m,\calC_{L_2})+l$.
Due to $d(m',\calC_{L_1}) \leq d(m,\calC_{L_2})$, we obtain
$d(m',\calC_3)\leq d(m,\calC_2)$.

\item
Consider any index $j\in S'$. Recall that $m'\leq i<j$ as $S'\subseteq
\calI[i+1,n]$. Therefore, $x^l_{m'}\leq x_j^l$. On the other
hand, $l_{m'}$ is to the right of $l_j$ in $\calC_3$. Thus, it holds that
$d(j,\calC_3)\leq d(m',\calC_3)$. We have proved above that
$d(m',\calC_3)\leq \delta_{opt}$. Hence, we also obtain
$d(j,\calC_3)\leq \delta_{opt}$.
\end{itemize}

This proves that $\calC_3$ is an optimal configuration. As discussed
above, the lemma follows.
\qed
\end{proof}

\begin{lemma}\label{lem:prune2}
Suppose $L_1$ and $L_2$ are two lists of $\calL$ whose last indices
are the same. Then, if $\delta(\calC_{L_1})\leq \delta(\calC_{L_2})$ and
$x(\calC_{L_1})\leq x(\calC_{L_2})$, then $L_1$ dominates $L_2$.
\end{lemma}
\begin{proof}
Assume $L_2$ is a canonical list of $\calI[1,i]$. Our goal is prove that $L_1$ is also
a canonical list of $\calI[1,i]$. To this end, it is sufficient to
construct an optimal configuration in which the order the intervals of
$\calI[1,i]$ is $L_1$. The proof techniques are similar to (but
simpler than) that for Lemma~\ref{lem:prune1}.

Let $m$ be the last index of $L_1$ and $L_2$.
Let $h$ (resp., $g$) be the left neighboring index of $m$ in $L_1$ (resp.,
$L_2$).

Since $L_2$ is a canonical list, there is an optimal list $Q$
that is consistent with $L_2$.  By the definition of $g$, all indices (if any)
strictly between $g$ and $m$ in $Q$ are from $\calI[i+1,n]$.
Let $S$ denote the set of indices of $\calI[i+1,n]$
before $g$ in $Q$. We consider the general case where $S\neq \emptyset$.

As in the proof of Lemma~\ref{lem:prune1},
we can obtain an optimal list $Q_1$ that is the same as $Q$
except that all indices of $S$ are now right after $g$ in $Q_1$
(i.e., all indices of $Q$ before $g$ except those in $S$ are
still before $g$ in $Q_1$ with the same relative order, and
all indices of $Q$ after $g$ are now after indices of $S$ in
$Q_1$ with the same relative order; e.g., see Fig.~\ref{fig:prune2}).
Therefore, in $Q_1$, the indices before and including $g$ are exactly those
in $\calI[1,i]\setminus\{m\}$.
%Also, the indices of $Q_1$
%before $j$ are exactly those in $L_1$ before $m$ (possibly with
%different order).

Let $Q_2$ be another list that is the same as $Q_1$ except
the following (e.g., see Fig.~\ref{fig:prune2}):
We re-arrange the indices before and including $g$ such that
they follow exactly the same order as in $L_1$. Note that
$L_1$  is consistent with $Q_2$. In the following, we show that $Q_2$  is an optimal list,
which will prove the lemma.

\begin{figure}[h]
\begin{minipage}[t]{\textwidth}
\begin{center}
\includegraphics[height=0.4in]{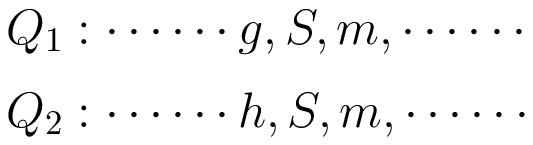}
\caption{\footnotesize Illustrating the two lists $Q_1$ and $Q_2$.
In $Q_1$ (resp., $Q_2$), the indices
strictly before $S$ are exactly those in
$\calI[1,i]\setminus\{m\}$.}
\label{fig:prune2}
\end{center}
\end{minipage}
\vspace{-0.15in}
\end{figure}

Since $Q_1$ is an optimal list, there is an optimal
configuration $\calC_1$ whose interval order is the same as $Q_1$. Consider
the configuration $\calC_2$ that is the same as $\calC_1$
except the following: For each interval $k$ before and
including $g$, we set the position of $I_k$ the same as its position in
$\calC_{L_1}$. Hence, the interval order of $\calC_2$ is the same as
$Q_2$. In the following, we show that $\calC_2$ is an optimal
configuration, which will prove that $Q_2$ is an optimal list.

We first show that $\calC_2$ is feasible. For each interval $k$ before
and including $h$, its position in $\calC_2$ is the same as that in
$\calC_{L_1}$, and thus interval $k$ is still valid in $\calC_2$. Other
intervals are also valid since they do not change their positions from
$\calC_1$ to $\calC_2$.  In the following, we show that no two
intervals overlap in $\calC_2$. Based on our way of constructing
$\calC_2$, it is sufficient to show that $x_{h}^r(\calC_2)\leq
x_t^l(\calC_2)$, where $t$ is the right neighboring index of $h$ in
$Q_2$. Note that $x_h^r(\calC_2) = x_h^r(\calC_{L_1})$
and $x_t^l(\calC_2)=x_t^l(\calC_1)$. In the following, we prove that
$x_h^r(\calC_{L_1}) \leq x_t^l(\calC_1)$. Depending on whether
$x_h^r(\calC_{L_1})\leq x_g^r(\calC_{L_2})$, there are two cases.

\begin{enumerate}
\item
If $x_h^r(\calC_{L_1})\leq x_g^r(\calC_{L_2})$, then since $L_2$ is
consistent with $Q_1$ and $\calC_{L_2}$ is constructed
based on the left-possible placement strategy, we have
$x_g^r(\calC_{L_2})\leq x_g^r(\calC_1)$, and thus, $x_h^r(\calC_{L_1})\leq x_g^r(\calC_1)$.

On the other hand, note that $t$ is also the right neighboring index
of $g$ in $Q_1$. Since $\calC_1$ is feasible,
$x_g^r(\calC_1)\leq x_t^l(\calC_1)$. Thus, we obtain
$x_h^r(\calC_{L_1}) \leq x_t^l(\calC_1)$.

\item
Assume  $x_h^r(\calC_{L_1})> x_g^r(\calC_{L_2})$. By the lemma
condition, we have $x_m^r(\calC_{L_1})=x(\calC_{L_1})\leq
x(\calC_{L_2})=x_m^r(\calC_{L_2})$. Since $x_h^r(\calC_{L_1})>
x_g^r(\calC_{L_2})$ and both $\calC_{L_1}$
and $\calC_{L_2}$ are constructed by the left-possible placement
strategy, it must be that
$x_m^l(\calC_{L_1})=x_m^l(\calC_{L_2})=x_m^l$, i.e., the positions of
$I_m$ in both $\calC_{L_1}$ and $\calC_{L_2}$ are the same as that in
the input.

Since $t$ is in $\calI[i+1,n]$ and $m\leq i$, $x_m^l\leq x_{t}^l$.
Since $x_t^l\leq x_t^l(\calC_{L_1})\leq x_t^l(\calC_1)$, it holds that
$x_m^l\leq x_t^l(\calC_1)$. Since $I_m$ is to the right of $I_h$
in the configuration $\calC_{L_1}$, $x_h^r(\calC_{L_1}) \leq x_m^l(\calC_{L_1})=x_m^l$.
Consequently, we obtain $x_h^r(\calC_{L_1}) \leq x_t^l(\calC_1)$.
\end{enumerate}

This proves that $\calC_2$ is feasible. In the sequel we show that
$\calC_2$ is an optimal configuration by proving that
$\delta(\calC_2)\leq \delta(\calC_1)= \delta_{opt}$. Since the intervals strictly
after $g$ do not change their positions from $\calC_1$ to $\calC_2$,
it is sufficient to show that $d(k,\calC_2)\leq \delta_{opt}$ for any
index $k$ before and including $g$ in $\calC_2$.

Since $x_k^l(\calC_2)=x_k^l(\calC_{L_1})$,
$d(k,\calC_2)=d(k,\calC_{L_1})\leq \delta(\calC_{L_1})$. By lemma
condition, $\delta(\calC_{L_1})\leq \delta(\calC_{L_2})$. Since $L_2$ is
consistent with $Q_1$ and $\calC_{L_2}$ is constructed
based on the left-possible placement strategy, it holds that
$\delta(\calC_{L_2})\leq \delta(\calC_{1})= \delta_{opt}$.
Combining the above discussions, we obtain $d(k,\calC_2)\leq
\delta(\calC_{L_1})\leq  \delta(\calC_{L_2}) \leq \delta_{opt}$.

This proves that $\calC_2$ is an optimal configuration. The lemma thus
follows.
\qed
\end{proof}

Let $E(\calL)$ denote the set of last intervals of all lists of $\calL$.
Our preliminary algorithm guarantees the following property on
$E(\calL)$, which will be useful later for our pruning algorithm given in
Section~\ref{sec:prune}.

\begin{lemma}\label{lem:endindex}
$E(\calL)$ has at most two intervals. Further, if $|E(\calL)|=2$, then
one interval of $E(\calL)$ contains the other one in the input.
\end{lemma}
\begin{proof}
We prove the lemma by induction. Initially, after $I_1$ is processed,
$\calL$ consists of the only list $L=\{1\}$.
Therefore, $E(\calL)=\{1\}$ and the lemma trivially holds.

We assume that the lemma holds after interval $I_{i-1}$ is processed. Let $\calL$ be the set after $I_{i}$ is processed. For differentiation, we let $\calL'$ denote the set $\calL$ before $I_{i}$ is processed.

Depending on whether the size of $E(\calL')$ is $1$ or $2$, there are two cases.

\paragraph{The case $|E(\calL')|=1$.} Let $m$ be the only index of
$E(\calL')$. Hence, for each list $L\in \calL'$, $m$ is the last index
of $L$. Depending on whether $x_m^r\leq x_i^r$, there are two subcases.

\begin{enumerate}
\item
If $x_m^r\leq x_{i}^{r}$, then according to our preliminary algorithm,
Case I of the algorithm happens on every list $L\in \calL'$, and
$i$ is appended at the end of $L$ for each $L\in \calL'$. Therefore,
the last indices of all lists of $\calL$ are $i$, and the lemma
statement holds for $E(\calL)$.

\item
If $x_m^r> x_{i}^{r}$, then note that $I_{i}\subseteq I_{m}$ in the input.
Consider any list $L\in \calL'$.
According to our preliminary algorithm, if $x_{i}^l\leq
x_m^l(\calC_{L})$, then $i$ is inserted into $L$ right before $m$;
otherwise, $i$ is appended at the end of $L$, and further, a new list
$L^*$ is produced in which $m$ is at the end.

Therefore, in this case, $E(\calL)$ has either one index or two
indices. If $|E(\calL)|=2$, then $E(\calL)=\{i,m\}$. Since
$I_{i}\subseteq I_{m}$ in the input, the lemma statement holds on
$E(\calL)$.
\end{enumerate}

\paragraph{The case $|E(\calL')|=2$.} By induction
hypothesis, one interval of $E(\calL')$ contains the other one in the
input. Let $m$ and $m'$ be the two indices of $E(\calL')$,
respectively, such that $I_{m'}\subseteq I_m$ in the input.
Hence, we have $m<m'$ and $x_{m'}^r\leq x_m^r$.

Depending on the $x$-coordinates of right endpoints of $I_{i}$,
$I_m$, and $I_{m'}$ in the input, there are three subcases: $x_m^r\leq
x_{i}^r$, $x_{m'}^r\leq x_{i}^r<x_m^r$, and $x_{i}^r<x_{m'}^r$.

\begin{enumerate}
\item
If $x_m^r\leq x_{i}^r$, then for each list $L\in \calL'$, Case I of
the algorithm
happens, and $i$ is appended at the end of $L$. Therefore, the
last indices of all lists of $\calL$ are $i$, and the lemma
statement holds for $E(\calL)$.

\item
If $x_{m'}^r\leq x_{i}^r<x_m^r$, then consider any list  $L\in
\calL'$. If $m'$ is at the end of $L$, then Case I
happens and $i$ is appended at the end of $L$.
If $m$ is at the end of $L$, then either Case II or Case III of the
algorithm happens. Hence, either $i$ or $m$ will be the last index
of $L$; if a new list $L^*$ is produced in Case III, then its
last index is $m$.

Therefore, after every list of $\calL'$ is processed, the last index
of each list of $\calL$ is either $m$ or $i$, i.e.,
$E(\calL)=\{m,i\}$.
%Since $m<i$, $x_m^l\leq x_{i}^l$. Due to $x_{i}^r<x_m^r$,
Note that $I_{i}$ is contained in $I_m$ in the input. Hence, the lemma statement
holds for $E(\calL)$.

\item
If $x_{i}^r<x_{m'}^r$, then $I_{i}$ is contained in both $I_m$ and
$I_{m'}$ in the input. Consider any list $L\in \calL'$. Regardless of
whether the last index is $m$ or $m'$, Case I does not happen.

We claim that Case III does not happen either. We prove the claim only for
the case where the last index of $L$ is $m$ (the other case can be
proved similarly). Indeed, in the configuration $\calC_L$, it holds
that $x_{m'}^r\leq x_{m'}^r(\calC_L)$. Since $m$ is the last index of
$L$, we have $x_{m'}^r(\calC_L)\leq x_m^l(\calC_L)$. Since
$x_{i}^r<x_{m'}^r$, we obtain $x_{i}^l\leq
x_{i}^r< x_{m'}^r\leq x_{m'}^r(\calC_L)\leq   x_m^l(\calC_L)$. This implies that Case
III of the algorithm cannot happen.

Hence, Case II happens, and $i$ is inserted into $L$ right before
the last index. Therefore, the last indices of all lists of $\calL$
are either $m$ or $m'$. The lemma statement holds for $E(\calL)$.
\end{enumerate}

This proves the lemma.
\qed
\end{proof}

\subsection{A Pruning Procedure}
\label{sec:prune}

Based on Lemmas~\ref{lem:prune1} and \ref{lem:prune2}, we present an
algorithm that prunes redundant lists from $\calL$ after each step for
processing an interval $I_{i}$. In the following, we describe the
algorithm, whose implementation is discussed in
Section~\ref{sec:imple}.

By Lemma~\ref{lem:endindex}, $E(\calL)$ has at most two indices. If
$E(\calL)$ has two indices, we let $m$ and $m'$ denote the two indices,
respectively, such that $I_{m'}\subseteq I_m$ in the input.
If $E(\calL)$ has only one index, let $m$ denote it and
$m'$ is undefined. Let $\calL_1$ (resp., $\calL_2$) denote the set of
lists of $\calL$ whose last indices are $m'$ (resp., $m$), and
$\calL_1=\emptyset$ if and only if $m'$ is undefined.

%To implement the algorithm in $O(n\log n)$ time,
Our algorithm maintains several invariants regarding certain monotonicity properties,
as follows, which are crucial to our efficient implementation.
%Our algorithm maintains the following invariants.

\begin{enumerate}
\item
$\calL$ contains a canonical list of $\calI[1,i]$.

\item
For any two lists $L_1$ and $L_2$ of $\calL$, $x(\calC_{L_1})\neq x(\calC_{L_2})$ and $\delta(\calC_{L_1})\neq \delta(\calC_{L_2})$.

\item
If $\calL_1\neq \emptyset$, then for any lists $L_1\in \calL_1$ and $L_2\in \calL_2$, $x(\calC_{L_1})<x(\calC_{L_2})$.

\item
For any two lists $L_1$ and $L_2$ of $\calL$, $x(\calC_{L_1})<
x(\calC_{L_2})$ if and only if
$\delta(\calC_{L_1})>\delta(\calC_{L_2})$. In other words, if we order
the lists $L$ of $\calL$ increasingly by the values $x(\calC_{L})$,
then the values $\delta(\calC_{L})$ are sorted decreasingly.
\end{enumerate}

After $I_n$ is processed, by the algorithm invariants, if $L$ is the
list of $\calL$ with minimum $\delta(\calC_L)$, then $L$ is an optimal list and
$\delta_{opt}=\delta(\calC_L)$.

Initially after the first interval $I_1$ is processed, $\calL$ has
only one list $L=\{1\}$, and thus, all algorithm invariants trivially
hold. In general, suppose the first $i-1$ intervals have been processed
and all algorithm invariants hold on $\calL$. In the following, we
discuss the general step for processing interval $I_{i}$.
%It is the same as the preliminary algorithm except
%that we will prune certain reduandan.  %after the set $\calL$ is updated.
%The details are given below.

For differentiation, we let $\calL'$ refer to the original set $\calL$
before interval $i$ is processed. Similarly, we use $\calL'_1$ and
$\calL'_2$ to refer to $\calL_1$ and $\calL_2$, respectively.
Let $L'_1,L'_2,\ldots,L'_{a}$ be the lists of $\calL'$ sorted with
$x(\calC_{L'_1})<x(\calC_{L'_2})<\cdots<x(\calC_{L'_a})$, where
$a=|\calL'|$. By the third invariant, we have
$\delta(\calC_{L'_1})>\delta(\calC_{L'_2})>\cdots>\delta(\calC_{L'_a})$.
If $\calL'_1= \emptyset$, let $b=0$; otherwise, let $b$ be the
largest index such that $L'_b\in \calL'_1$, and by the third
algorithm invariant, $\calL_1'=\{L_1',\ldots,L_b'\}$ and
$\calL_2'=\{L_{b+1}',\ldots,L_a\}$.
Depending on whether $\calL'_1=\emptyset$, there are two main
cases.

\subsubsection{The Case $\calL'_1=\emptyset$}
\label{sec:mainfirst}

In this case, for each list $L'\in \calL'$, its last index is $m$.
Depending on whether $x_{m}^r\leq x_{i}^r$, there are two
subcases.

\paragraph{The first subcase $x_{m}^r\leq x_{i}^r$.} In this case,
according to the preliminary
algorithm, for each list $L_j'\in \calL'$, Case I happens and $i$
is appended at the end of $L'_j$, and we use $L_j$ to refer to the updated list of $L_j'$
with $i$. According to our left-possible placement strategy,
$x_{i}^l(\calC_{L_j})=\max\{x(\calC_{L'_j}),x_{i}^l\}$.
Thus, $x(\calC_{L_j})=x_{i}^l(\calC_{L_j})+|I_i|$ and
$d(i,\calC_{L_j})=x_{i}^l(\calC_{L_j}) - x_{i}^l$.

As the index $j$ increases from $1$ to $a$, since the value $x(\calC_{L'_j})$ strictly
increases, $x_{i}^l(\calC_{L_j})$ (and thus $x(\calC_{L_j})$ and
$d(i,\calC_{L_j})$) is monotonically increasing (it may first
be constant and then strictly increases after some index, say, $a_1$).
Formally, we define $a_1$ as follows.
If $x(\calC_{L'_1})> x_{i}^l$, then let $a_1=0$; otherwise,
define $a_1$ to be the largest index $j\in [1,a]$ such that
$x(\calC_{L'_j})\leq x_{i}^l$ (e.g., see Fig.~\ref{fig:defa1}).
In the following, we first assume $a_1\neq 0$. As discussed above, as $j$ increases
in $[1,a]$, $x_{i}^l(\calC_{L_j})$ is constant on $j\in [1,a_1]$ and strictly
increases on $j\in [a_1,a]$.

\begin{figure}[t]
\begin{minipage}[t]{\textwidth}
\begin{center}
\includegraphics[height=0.9in]{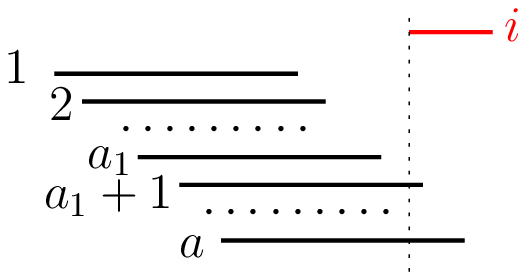}
\caption{\footnotesize Illustrating the definition of $a_1$. The black segments show the positions
of interval $m$ in the configurations $\calC_{L_j'}$ for $j\in [1,a]$, and the
numbers on the left side are the
indices of the lists. The red segment shows the interval $i$ in the input position.}
\label{fig:defa1}
\end{center}
\end{minipage}
\vspace{-0.15in}
\end{figure}

Now consider the value $\delta(\calC_{L_j})$, which is equal to
$\max\{\delta(\calC_{L'_j}),d(i,\calC_{L_j})\}$ by
Observation~\ref{obser:case1}. Recall that
$\delta(\calC_{L'_j})$ is strictly decreasing on $j\in [1,a]$. Observe that $d(i,\calC_{L_j})$ is $0$ on $j\in [1,a_1]$ and
strictly increases on $j\in [a_1,a]$. This implies that
$\delta(\calC_{L_j})$ on $j\in [1,a]$ is a  unimodal function, i.e.,
it first strictly decreases and then strictly increases after some
index, say, $a_2$. Formally, let $a_2$ be the largest index $j\in [a_1+1,a]$ such that
$\delta(\calC_{L_{j-1}})> \delta(\calC_{L_{j}})$,
and if no such index $j$ exists, then let $a_2=a_1$.
The following lemma is proved based on Lemma~\ref{lem:prune2}.
%By the lemma, if $a_1>1$, then the lists
%$L_1,L_2,\ldots,L_{a_1-1}$ can all be pruned from $\calL$.

\begin{lemma}\label{lem:110}
\begin{enumerate}
\item
If $a_1>1$, then for each $j\in [1,a_1-1]$, $L_{a_1}$ dominates $L_j$.
\item
If $a_2<a$, then for each $j\in [a_2+1,a]$, $L_{a_2}$ dominates $L_j$.
\end{enumerate}
\end{lemma}
\begin{proof}
\begin{enumerate}
\item
Let $k=a_1$ and assume $k>1$. Consider any $j\in [1,k-1]$.
%According to our preliminary algorithm,
By the definition of $a_1$, $x_{i}^l(\calC_{L_j})=x_{i}^l(\calC_{L_k})=x_{i}^l$. Therefore,
$x(\calC_{L_j})=x(\calC_{L_k})=x_{i}^l+|I_{i}|$. Since
$d(i,\calC_{L_j})=d(i,\calC_{L_k})=0$, we have
$\delta(\calC_{L_j})=\delta(\calC_{L'_j})$ and
$\delta(\calC_{L_k})=\delta(\calC_{L'_k})$.
Since $j<k$, $\delta(\calC_{L'_j})>\delta(\calC_{L'_k})$. Thus,
we obtain $\delta(\calC_{L_j})>\delta(\calC_{L_k})$.

Since $x(\calC_{L_j})=x(\calC_{L_k})$,
$\delta(\calC_{L_j})>\delta(\calC_{L_k})$, and the last indices of $L_j$
and $L_k$ are both $i$,
by Lemma~\ref{lem:prune2}, $L_k$ dominates $L_j$.

\item
Let $k=a_2$ and assume $k<a$. Consider any $j\in [k+1,a]$.
As discussed before, $x(\calC_{L_j})$ is monotonically increasing on
$j\in [1,a]$. Thus, $x(\calC_{L_k})\leq x(\calC_{L_j})$.
By the definition of $a_2$ and since $\delta(\calC_{L_j})$ is a
unimodal function on $j\in [1,a]$, it holds that
$\delta(\calC_{L_{k}})\leq \delta(\calC_{L_{j}})$.
%Recall that $i$ is the last index of both $L_k$ and $L_j$.
By Lemma~\ref{lem:prune2}, $L_k$ dominates $L_j$.
\end{enumerate}
This proves the lemma.
\qed
\end{proof}

By Lemma~\ref{lem:110},
we let $\calL=\{L_j\ |\  a_1 \leq j\leq a_2\}$. The above is for the
general case where $a_1\neq 0$. If $a_1=0$, then we let $\calL=\{L_j\
|\  1 \leq j\leq a_2\}$.

%(we ignore $L_0$ if $a_1=0$).

\begin{observation}\label{obser:40}
All algorithm invariants hold for $\calL$.
\end{observation}
\begin{proof}
By Lemma~\ref{lem:110}, the lists that have been removed are redundant. Hence, $\calL$ contains a canonical list of $\calI[1,i]$ and the first algorithm invariant holds.

By our definitions of $a_1$ and $a_2$, when $j$ increases in $[a_1,a_2]$, $x(\calC_{L_j})$ strictly increases and $\delta(\calC_{L_j})$ strictly decreases. Therefore, the last three algorithm invariants hold.
\qed
\end{proof}

The following lemma will be quite useful for the algorithm implementation
given later in Section~\ref{sec:imple}.
%Intuitively, our pruning
%procedure guarantees the properties in the lemma, which in turn lead
%to an efficient algorithm implementation.

\begin{lemma}\label{lem:case11}
If $a_1<a_2$, then for each $j\in [a_1+1,a_2]$,
$x(\calC_{L_j})=x(\calC_{L'_j})+|I_{i}|$. For each list $L_j\in \calL$ with $j\neq a_2$, $\delta(\calC_{L_j})=\delta(\calC_{L'_j})$.
\end{lemma}
\begin{proof}
By the definition of $a_1$, for any $j\in [a_1+1,a]$, it always holds  that
$x(\calC_{L_j})=x(\calC_{L'_j})+|I_{i}|$. This proves the first lemma statement.

Recall that $\delta(\calC_{L_j})=\max\{\delta(\calC_{L'_j}),d(i,\calC_{L_j})\}$ for each $j\in [1,a]$.

Consider any list $L_j$ with $j\neq a_2$. Assume to the contrary that
$\delta(\calC_{L_j})\neq \delta(\calC_{L'_j})$. Then,
$\delta(\calC_{L_j})=d(i,\calC_{L_j})$. Since
$\delta(\calC_{L_j})=d(i,\calC_{L_j})<d(i,\calC_{L_{a_2}})$,
we obtain  $\delta(\calC_{L_j})\leq \delta(\calC_{L_{a_2}})$, which
contradicts with $\delta(\calC_{L_j}) > \delta(\calC_{L_{a_2}})$.
\qed
\end{proof}

\paragraph{The second subcase $x_{m}^r > x_{i}^r$.}
In this case, for each list $L_j'\in \calL'$, according to our
preliminary algorithm, depending on whether
$x_{i}^l\leq x_{m}^l(\calC_{L'_j})$, either Case II or Case III can
happen. If $x_{i}^l\leq x_{m}^l(\calC_{L'_1})$, then let $c=0$; otherwise,
let $c$ be the largest index $j$ such that  $x_{i}^l >
x_{m}^l(\calC_{L'_j})$ (e.g., see Fig.~\ref{fig:defc}). In the following, we first consider the general case where $1\leq c<a$.

\begin{figure}[t]
\begin{minipage}[t]{\textwidth}
\begin{center}
\includegraphics[height=0.9in]{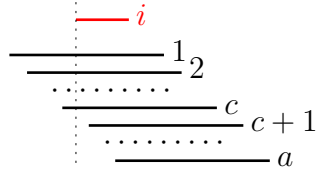}
\caption{\footnotesize Illustrating the definition of $c$. The black segments show the positions
of interval $m$ in the configurations $\calC_{L_j'}$ for $j\in [1,a]$, and the numbers on the right side are the
indices of the lists. The red segment shows the interval $i$ in the input position.}
\label{fig:defc}
\end{center}
\end{minipage}
\vspace{-0.15in}
\end{figure}

For each $j\in [1,c]$, observe that
$x_m^l(\calC_{L'_j})=x(\calC_{L'_j})-|I_m|\leq x(\calC_{L'_c})-|I_m|=x_m^l(\calC_{L'_c})<x_i^l$.
According to our preliminary algorithm, Case III happens, and thus $L_j'$ will
produce two lists: the list $L_j$ by appending $i$ at the end of
$L_j'$, and the new list $L_j^*$ by inserting $i$ in front of $m$
in $L_j'$. Further, %since $I_{i}$ is contained in $I_m$ in the input,
according to our left-possible placement strategy,
$x^l_{i}(\calC_{L_j})=x(\calC_{L_j'})$ in $\calC_{L_j}$, and
$x^l_{i}(\calC_{L^*_j})=x^l_{i}$ and
$x^l_{m}(\calC_{L^*_j})=x^r_{i}$ in $\calC_{L^*_j}$.
By Observation~\ref{obser:case3},
$\delta(\calC_{L_j})=\max\{\delta(\calC_{L_j'}),d(i,\calC_{L_j})\}$
and
$\delta(\calC_{L^*_j})=\max\{\delta(\calC_{L_j'}),d(m,\calC_{L^*_j})\}$.

\begin{observation}\label{obser:50}
$\delta(\calC_{L^*_c})\leq \delta(\calC_{L^*_j})$ for any $j\in
[1,c]$.
\end{observation}
\begin{proof}
For any $j\in [1,c]$, note that
$d(m,\calC_{L^*_j})=x^l_{m}(\calC_{L^*_j}) - x_{m}^l = x_{i}^r-x_{m}^l$. Therefore,
$d(m,\calC_{L^*_j})$ is the same for all $j\in [1,c]$. On the other
hand, we have $\delta(\calC_{L'_j})\geq \delta(\calC_{L'_c})$. Thus,
$\delta(\calC_{L^*_c})\leq \delta(\calC_{L^*_j})$. \qed
\end{proof}

By the above observation and Lemma~\ref{lem:prune},  among the new lists $L_j^*$ with
$j=1,2,\ldots,c$, only $L_c^*$ needs to be kept. %in $\calL$.

For each $j\in[1,c]$, note that
$x(\calC_{L_j})=x(\calC_{L_j'})+|I_i|$. Since $x(\calC_{L_j'})$ is
strictly increasing on $j\in[1,c]$, $x(\calC_{L_j})$ is also
strictly increasing on $j\in[1,c]$.  Since
$d(i,\calC_{L_j})=x^l_{i}(\calC_{L_j})-x_{i}^l=x(\calC_{L_j'})-x_{i}^l$
for any $j\in [1,c]$, $d(i,\calC_{L_j})$ also strictly increases  on $j\in [1,c]$.
Further, since $\delta(\calC_{L'_j})$ strictly decreases on $j\in [1,c]$,
$\delta(\calC_{L_j})$, which is equal to $\max\{\delta(\calC_{L_j'}),d(i,\calC_{L_j})\}$,
is a unimodal function (i.e., it first strictly decreases and then strictly increases).
%We consider the sequence of lists $L_1\ldots,L_c$.
Let $c_1$ be the smallest index $j\in [1,c-1]$ such that $\delta(\calC_{L_j})\leq
\delta(\calC_{L_{j+1}})$, and if such an index $j$ does not exist, then let $c_1=c$.

\begin{lemma}\label{lem:130}
If $c_1<c$, then $L_{c_1}$ dominates $L_j$ for any $j\in [c_1+1,c]$.
\end{lemma}
\begin{proof}
Consider any $j\in [c_1+1,c]$. Since $\delta(\calC_{L_j})$ is a
unimodal function on $j\in [1,c]$, by the definition of $c_1$,
$\delta(\calC_{L_{c_1}})\leq \delta(\calC_{L_j})$. Recall that
$x(\calC_{L_{c_1}})\leq x(\calC_{L_j})$.
Since the last indices of $L_{c_1}$ and $L_j$ are both $i$, by
Lemma~\ref{lem:prune2}, $L_{c_1}$ dominates $L_j$.  \qed
\end{proof}

By the preceding lemma, if $c_1<c$, then we do not have to keep the
lists $L_{c_1+1},\ldots,L_c$ in $\calL$. Let
$S_1=\{L_1,\ldots,L_{c_1}\}$.

Consider any index $j\in [c+1,a]$. By the definition of $c$ and also
due to that $x(\calC_{L_k'})$ is strictly increasing on $k\in [1,a]$,
it holds that $x_m^l(\calC_{L_j'})\geq x_i^l$, and thus
Case II of the preliminary algorithm happens on $L_j'$ and
$L_j$ is obtained by inserting $i$ right before $m$ in $L_j'$.
%Further, $x_{i}^l(\calC_{L_j})=x_{m}^l(\calC_{L'_j})$ and
%$x_{m}^l(\calC_{L_j})=x_{i}^l(\calC_{L'_j})+|I_{i}|$.
%Thus, $d(m,\calC_{L_j})=d(m,\calC_{L'_j})+|I_{i}|$.
%Since $I_{i}$ is contained in $I_{m}$ in the input, it is easy to see that
%$d(m,\calC_{L_j})> d(i,\calC_{L_j})$.  Therefore,
By Observation~\ref{obser:case2},
$\delta(\calC_{L_j})=\max\{\delta(\calC_{L'_j}),d(m,\calC_{L_j})\}$.
Note that $x(\calC_{L_j})=x(\calC_{L'_j})+|I_{i}|$ and
$x_{m}^r(\calC_{L_j})=x(\calC_{L_j})$.
%Hence, $d(m,\calC_{L_j})=x(\calC_{L'_j})+len(I_{i})-x_{m}^r$.
As $j$ increases in $[c+1,a]$, since $x(\calC_{L_j'})$ strictly
increases, both $x(\calC_{L_j})$ and $d(m,\calC_{L_j})$ strictly increase.
Since $\delta(\calC_{L'_j})$ is strictly decreasing on $j\in [c+1,a]$, we
obtain that $\delta(\calC_{L_j})$  is a unimodal function on $j\in
[c+1,a]$ (i.e., it first strictly decreases and then strictly
increases).

Let $S= \{L_1,\ldots,L_c,L_c^*,L_{c+1},\ldots,L_a\}$.
For convenience, we use $L_{c+0.5}$ to refer to $L_c^*$
(and $L'_{c+0.5}$ refers to $L'_c$);
in this way, the indices of the ordered lists of $S$ are sorted.
%Recall that $S= \{L_1,\ldots,L_c,L_c^*,L_{c+1},\ldots,L_a\}$ with $L_{c+0.5}=L^*_c$.
Consider the subsequence of the lists of $S$ from $L_{c+0.5}$ to the end
(including $L_{c+0.5}$).
Define $c_2$ to be the index of the first list $L_j$ such that
$\delta(\calC_{L_{j}})\leq \delta(\calC_{L})$, where $L$ is the right
neighboring list of $L_{j}$ in $S$; if such a list $L_j$ does not exist,
then we let $c_2=a$.

\begin{observation}\label{obser:60}
As $j$ increases in $[1,a]$, $x(\calC_{L_j})$ is strictly increasing
except that $x(\calC_{L_{c+0.5}})= x(\calC_{L_{c+1}})$ may be possible.
\end{observation}
\begin{proof}
Recall that $x(\calC_{L_j})$ is strictly increasing on $j\in [1,c]$ and
$j\in [c+1,a]$, respectively.
Let $l=|I_{i}|+|I_{m}|$.
Note that $x(\calC_{L_c})=x_{m}^l(\calC_{L'_c})+l$,
$x(\calC_{L^*_c})=x^l_{i}+l$, and
$x(\calC_{L_{c+1}})=x^l_{m}(\calC_{L'_{c+1}})+l$.
By our definition of $c$, $x^l_{m}(\calC_{L'_{c}})< x^l_{i}\leq
x^l_{m}(\calC_{L'_{c+1}})$.
Thus, $x(\calC_{L_c})<x(\calC_{L^*_c})\leq x(\calC_{L_{c+1}})$. This
shows that $x(\calC_{L_j})$ is strictly increasing on $j\in [1,a]$
except that $x(\calC_{L^*_c})= x(\calC_{L_{c+1}})$ may be possible.
\qed
\end{proof}

\begin{lemma}\label{lem:140}
\begin{enumerate}
\item
If $c_2<a$, then $L_{c_2}$ dominates $L_j$ for any $L_j\in S$
with $j>c_2$.
\item
If $c_2\geq c+1$ and $x(\calC_{L_{c+0.5}})= x(\calC_{L_{c+1}})$, then
$L_{c+1}$ dominates $L_{c+0.5}$.
\end{enumerate}
\end{lemma}
\begin{proof}
We first show that $\delta(\calC_{L_j})$ is a unimodal function on
$j\in [c+0.5,a]$.

Recall that
%for each $j\in [1,c]$,
%$\delta(\calC_{L_j})=\max\{\delta(\calC_{L'_j}),d(i,\calC_{L_j})\}$,
%and $d(i,\calC_{L_j})$ is strictly increasing on $j\in [1,c]$.
for each $j\in [c+1,a]$, $\delta(\calC_{L_j})=\max\{\delta(\calC_{L'_j}),d(m,\calC_{L_j})\}$,
and
$\delta(\calC_{L^*_j})=\max\{\delta(\calC_{L'_j}),d(m,\calC_{L^*_j})\}$.
For each $j\in [c+0.5,a]$, since $m$ is the last index of $L_j$,
we have $d(m,\calC_{L_{j}})=x(\calC_{L_j})-x_m^r$. %=x(\calC_{L_j})+|I_{i}|-x_m^r$.
By Observation~\ref{obser:60},
%Note that
%$\delta(\calC_{L^*_j})=\max\{\delta(\calC_{L'_j}),d(m,\calC_{L^*_j})\}$.
%We claim that $d(m,\calC_{L^*_{c}})\leq d(m,\calC_{L_{c+1}})$.
%Indeed, $d(m,\calC_{L^*_{c}})=x_{i}^l+l-x_{m}^l$ and
%$d(m,\calC_{L_{c+1}})=x_{m}^l(\calC_{L'_{c+1}})+l-x_{m}^l$. By
%the definition of $c$, $x^l_{i}\leq x^l_{m}(\calC_{L'_{c+1}})$,
%and thus, $d(m,\calC_{L^*_{c}})\leq d(m,\calC_{L_{c+1}})$.
%For each $j\in [c+1,a]$, since $m$ is the last index of $L_j$,
%we also have
%$d(m,\calC_{L_{c+1}})=x(\calC_{L_j})-x_m^r=x(\calC_{L_j})+len(I_{i})-x_m^r$.
%As $x(\calC_{L'_j})$ on $j\in [c+1,a]$ is strictly increasing,
%$d(m,\calC_{L_j})$ is strictly increasing on $j\in [c+1,a]$.
%
%Combining the above discussions, we obtain that
$d(m,\calC_{L_{j}})$ is strictly increasing on $[c+0.5,a]$ except that
$d(m,\calC_{L_{c+0.5}}) = d(m,\calC_{L_{c+1}})$ may be possible.
Since $\delta(\calC_{L'_j})$ on $j\in
[1,a]$ is strictly decreasing, $\delta(\calC_{L_j})$ is a unimodal
function on $j\in [c+0.5,a]$.

By the definition of $c_2$, $\delta(\calC_{L_j})$
is strictly decreasing on $[c+0.5,c_2]$ and monotonically increasing
on $[c_2,a]$.

Consider any list $L_j\in S$ with $j>c_2$. By our previous discussion,
$\delta(\calC_{L_{c_2}})\leq \delta(\calC_{L_j})$ and
$x(\calC_{L_{c_2}})\leq x(\calC_{L_j})$.
Since the last indices of both $L_{c_2}$ and $L_j$ are $m$, by
Lemma~\ref{lem:prune2}, $L_{c_2}$ dominates $L_j$.

If $c_2\geq c+1$ and $x(\calC_{L_{c+0.5}})= x(\calC_{L_{c+1}})$, by the
definition of $c_2$, $\delta(\calC_{L_{c+0.5}})>
\delta(\calC_{L_{c+1}})$. Since the last indices of both $L_{c+0.5}$
and $L_{c+1}$ are $m$, by Lemma~\ref{lem:prune2}, $L_{c+1}$
dominates $L_{c+0.5}$.
The lemma thus follows. \qed
\end{proof}

Let $S_2=\{L_{c+0.5},L_{c+1},\ldots,L_{c_2}\}$ and we remove
$L_{c+0.5}$ from $S_2$ if $c_2\geq c+1$ and
$x(\calC_{L_{c+0.5}})=x(\calC_{L_{c+1}})$. In the following, we combine
$S_1$ and $S_2$ to obtain the set $\calL$. We consider the lists of $S_2$ in
order. Define $c'$ to be the index $j$ of the first list $L_j$ such that
$\delta(\calC_{L_{c_1}})>\delta(\calC_{L_j})$, and if no such list
$L_j$ exists, then let $c'=c_2+1$.

\begin{lemma}\label{lem:150}
If $L_{c'}$ is not the first list of $S_2$ or $c'=c_2+1$, then for each list $L_j$
of $S_2$ with $j<c'$, $L_{c_1}$ dominates $L_j$.
\end{lemma}
\begin{proof}
We assume that $L_{c'}$ is not the first list of $S_2$ or $c'=c_2+1$.

Note that we have proved in the proof of Lemma~\ref{lem:140} that
$\delta(\calC_{L_j})$ on $j\in [c+0.5,c_2]$ is strictly decreasing.
By the
definition of $c'$, it holds that
$\delta(\calC_{L_{c_1}})\leq \delta(\calC_{L_j})$ for any $L_j\in S_2$
with $j<c'$.

Consider any list $L_j$ of $S_2$ with $j<c'$.

Recall that $\delta(\calC_{L_j})=\max\{\delta(\calC_{L'_j}),d(m,\calC_{L_j})\}$.
We claim that $\delta(\calC_{L_j})=d(m,\calC_{L_j})$. Indeed,
note that %$\delta(\calC_{L'_{c_1}})\geq \delta(\calC_{L'_j})$ and
$\delta(\calC_{L_j'})\leq \delta(\calC_{L'_{c_1}})\leq  \delta(\calC_{L_{c_1}})$.
%Thus, $\delta(\calC_{L'_j}) \leq \delta(\calC_{L_{c_1}})$.
Since $\delta(\calC_{L_{c_1}})\leq \delta(\calC_{L_j})$, we obtain
$\delta(\calC_{L'_j}) \leq \delta(\calC_{L_j})$, and thus,
$\delta(\calC_{L_j})=d(m,\calC_{L_j})$.

Consequently, we have
$\delta(\calC_{L_{c_1}})\leq d(m,\calC_{L_j})$ and
$x(\calC_{L_{c_1}})\leq x(\calC_{L_j})$ (by Observation~\ref{obser:60}).
Further, the last index of $L_{c_1}$ is $i$ and the last index of
$L_j$ is $m$, with $x_{i}^r\leq x_m^r$.
By Lemma~\ref{lem:prune1}, $L_{c_1}$ dominates
$L_j$.

%If $j=c+0.5$, then the goal is to show that $L_{c_1}$ dominates
%$L_c^*$.  The proof is similar to the above second subcase.
%In this case, we have
%$\delta(\calC_{L^*_c})=\max\{\delta(\calC_{L'_c}),d(m,\calC_{L_c^*})\}$.
%By the same analysis as above, we can also show that
%$\delta(\calC_{L^*_c})=d(m,\calC_{L_c^*})$. Consequently, as above,
%we can apply Lemma~\ref{lem:prune1} to show that $L_{c_1}$ dominates
%$L_c^*$.

The lemma thus follows. \qed
\end{proof}

We remove from $S_2$ all lists $L_j$ with $j<c'$, and let
$\calL=S_1\cup S_2$. In general, if $c'\neq c_2+1$, then
$\calL=\{L_1,\ldots,L_{c_1},L_{c'},\ldots,L_{c_2}\}$;
otherwise, $\calL=\{L_1,\ldots,L_{c_1}\}$.
%Further, if $c_2\geq c+1$  and $x(\calC_{L_{c+0.5}})=
%x(\calC_{L_{c+1}})$, then we also remove $L^*_c$  from $\calL$.

The above discussion is for the general case where $1\leq c<a$.
If $c=0$, then $L_c^*$, $c_1$ and $c'$ are all undefined, and
we have $\calL=\{L_{1},\ldots,L_{c_2}\}$.
If $c=a$, then $\calL=\{L_1,\ldots,L_{c_1}\}$ if $\delta(L_{c_1})\leq
\delta(L_c^*)$ and $\calL=\{L_1,\ldots,L_{c_1},L_c^*\}$ otherwise.

\begin{observation}
All algorithm invariants hold on $\calL$.
\end{observation}
\begin{proof}
We only consider the most general case where $1\leq c<a$ and
$c'\neq c_2+1$, since other cases can be proved in a similar but
easier way.

By Lemmas~\ref{lem:130}, \ref{lem:140}, and \ref{lem:150},
all pruned lists are redundant and thus
$\calL$ contains a canonical list of $\calI[1,i]$. The first
algorithm invariant holds.

If $x(\calC_{L_{c+0.5}})=x(\calC_{L_{c+1}})$, then $L_{c+0.5}$ and
$L_{c+1}$ cannot be both in $\calL$ by Lemma~\ref{lem:140}(2). Thus,
by Observation~\ref{obser:60}, $x(\calC_{L_j})$ strictly increases in
$[1,a]$. Recall that for any list $L_j\in \calL$,
the last index of $L_j$ is $i$ if $j\leq c_1$ and $m$ otherwise.
Recall that $I_{i}$ is contained in $I_m$ in the input. Thus, the
fourth algorithm invariant holds.

Further, our definitions of $c_1$, $c'$, and $c_2$ guarantee that
$\delta(\calC_{L})$ on all lists $L$ following their order in $\calL$
is strictly decreasing.
Therefore, the other two algorithm invariants also hold.
\qed
\end{proof}

The following lemma will be useful for the algorithm implementation.

\begin{lemma}\label{lem:case12}
For each list $L_j\in \calL$, if $L_j\neq L_c^*$, then
$x(\calC_{L_j})=x(\calC_{L'_j})+|I_{i}|$; if $L_j\not\in \{L_c^*,
L_{c_1}, L_{c_2}\}$, then $\delta(\calC_{L_j})=\delta(\calC_{L'_j})$.
\end{lemma}
\begin{proof}
If $L_j\neq L_c^*$, then we have discussed before that
$x(\calC_{L_j})=x(\calC_{L'_j})+|I_{i}|$ always holds regardless of
whether the last index of $L_j$ is $i$ or $m$.

If $L_j\not\in \{L_c^*, L_{c_1}, L_{c_2}\}$, assume to the contrary
that $\delta(\calC_{L_j})\neq \delta(\calC_{L'_j})$. Then, since
$\delta(\calC_{L_j})=\max\{\delta(\calC_{L'_j}),d(k,\calC_{L_j})\}$,
we obtain that $\delta(\calC_{L_j})=d(k,\calC_{L_j})$,
where $k$ is the last index of $\calC_{L_j}$ ($k$ is $i$ if $j\leq
c$ and $m$ otherwise). Note that $j$ is either in $[1,c_1]$ or
$[c',c_2]$. We discuss the two cases below.

\begin{enumerate}
\item
If $j\in [1,c_1]$, then the last index of $L_j$ is $i$.
Since $L_j\neq L_{c_1}$, $j<c_1$ holds. We have discussed before that
$d(i,\calC_{L_j})\leq d(i,\calC_{L_{c_1}})$. Thus, we can deduce
$\delta(\calC_{L_j})=d(i,\calC_{L_j})\leq d(i,\calC_{L_{c_1}})\leq
\delta(\calC_{L_{c_1}})$. However, we have already proved that
$\delta(\calC_{L_j})>\delta(\calC_{L_{c_1}})$. Thus, we obtain
contradiction.

\item
If $j\in [c',c_2]$, the analysis is similar. In this case the
last index of $L_j$ is $m$ and $j<c_2$. Since
$j<c_2$, we have discussed before that $d(m,\calC_{L_j})\leq
d(m,\calC_{L_{c_2}})$. Thus, we can deduce
$\delta(\calC_{L_j})=d(m,\calC_{L_j})\leq d(m,\calC_{L_{c_2}})\leq
\delta(\calC_{L_{c_2}})$. However, we have already proved that
$\delta(\calC_{L_j})>\delta(\calC_{L_{c_2}})$. Thus, we obtain
contradiction.
\end{enumerate}

The lemma thus follows.
\qed
\end{proof}

\subsubsection{The Case $\calL_1'\neq \emptyset$}
\label{sec:mainsecond}

We then consider the case where $\calL_1'\neq \emptyset$. In this case, recall that
$\calL_1'=\{L_1',\ldots,L_b'\}$ and $\calL_2'=\{L_{b+1}',\ldots,L_a'\}$. For each $L'_j\in \calL'$, the last index of
$L'_j$ is $m'$ if $j\leq b$ and $m$ otherwise. Recall that
$I_{m'}\subseteq I_{m}$ in the input. As in the proof of
Lemma~\ref{lem:endindex}, there are three subcases:
$x_{i}^r\geq x_{m}^r$, $x_{m'}^r\leq x_{i}^r< x_{m}^r$, and
$x_{i}^r< x_{m'}^r$.

\paragraph{The first subcase $x_{i}^r\geq x_{m}^r$.}
In this case, for each $L_j'\in \calL'$, Case I of the preliminary
algorithm happens and $L_j$ is obtained by appending $i$ at
the end of $L_j'$. Our pruning procedure for this subcase is similar to the first subcase
in Section~\ref{sec:mainfirst}, and we briefly discuss it below.

First, for each $L_j'\in \calL'$, $x_{i}^l(\calC_{L_j})=\max\{x(\calC_{L'_j}),x_{i}^l\}$ and
$\delta(\calC_{L_j})=\max\{\delta(\calC_{L'_j}),d(i,\calC_{L_j})\}$.
We define $a_1$ and $a_2$ in exactly the same way as in the first subcase
of Section~\ref{sec:mainfirst}, and further, Lemma~\ref{lem:110} still
holds. Similarly, we let $\calL$ consist of only those lists $L_j$
with $j\in [a_1,a_2]$. By the similar analysis,
Observation~\ref{obser:40} and Lemma~\ref{lem:case11} still hold. We
omit the details.

\paragraph{The second subcase $x_{m'}^r\leq x_{i}^r< x_{m}^r$.}
In this case, we first apply the similar pruning procedure for the first
(resp., second) subcase of
Section~\ref{sec:mainfirst} to set $\calL_1'$ (resp., $\calL_2'$), and then
we combine the results. The details are given below.

For set $\calL_1'$, the last indices of all lists of $\calL_1'$ are
$m'$. Since $x_{m'}^r\leq x_{i}^r$, for each $L_j'\in \calL'_1$, Case
I of the preliminary algorithm happens and $L_j$ is obtained by appending $i$ at the end of $L_j'$. We define $a_1$ and $a_2$ in the similar way as in
the first subcase of Section~\ref{sec:mainfirst} but with respect to
the indices in $[1,b]$. In fact, since $x_{i}^r< x_{m}^r$,
it holds that $x_{i}^l\leq  x_{i}^r \leq x_{m}^r\leq
x(\calC_{L'_1})$, and consequently, $a_1=0$.
%Define $a_2$ to be the largest index $j\in [1,b]$ such that
%$\delta(\calC_{L_{j-1}})>\delta(\calC_{L_j})$.
Similarly, Lemma~\ref{lem:110} also holds with respect to the indices
of $[1,b]$. Further, as $j$
increases in $[1,a_2]$, $x(\calC_{L_j})$ is strictly increasing and
$\delta(\calC_{L_j})$ is strictly decreasing. Let $S'_1=\{L_1,L_2,\ldots,L_{a_2}\}$.

For set $\calL_2'$, the last indices of all its lists are
$m$. Since $x_{i}^r<x_{m}^r$, for each list $L_j'\in \calL_2$,
either Case II or Case III of the algorithm happens. We define $c$ in the similar way
as in the second subcase of Section~\ref{sec:mainfirst} but with
respect to the indices of $[b+1,a]$.
%(if $x_{i}^l\leq x_m^l(\calC_{L_{b+1}})$, then we let $c=b$ instead of $0$).
Specifically,
if $x_{i}^l\leq x_{m}^l(\calC_{L'_{b+1}})$, then let $c=b$; otherwise,
let $c$ be the largest index $j\in [b+1,a]$ such that  $x_{i}^l >
x_{m}^l(\calC_{L'_j})$. We consider the most
general case where $b+1\leq c<a$ (other cases are similar but easier).

For each $j\in [b+1,c]$, there is also a new list $L^*_j$. Similar to
Observation~\ref{obser:40}, $\delta(\calC_{L^*_c})\leq
\delta(\calC_{L^*_j})$ for any $j\in [b+1,c]$. Hence, among the new lists
$L_j^*$ with $j=b+1,\ldots,c$, only $L_c^*$ needs to be kept.
Let $S'=\{L_{b+1},\ldots,L_c,L_c^*,L_{c+1},\ldots,L_a\}$.
We also use $L_{c+0.5}$ to refer to $L_c^*$.
We define the three indices $c_1$, $c_2$, and $c'$ in the similar way as
in the second subcase of Section~\ref{sec:mainfirst} but with respect
to the ordered lists in $S'$.
Similarly, Observation~\ref{obser:60}, Lemmas~\ref{lem:130},
\ref{lem:140}, and \ref{lem:150} all hold with respect to the lists in
$S'$. Let
$S'_2=\{L_{b+1},\ldots,L_{c_1},L_{c'},\ldots,L_{c_2}\}$.

%By the preceding lemma, for each $j>c'$, we  remove $L_j$ from $S_2$, and thus $S_2=\{L_{b+1},L_{b+2},\ldots,L_{c'}\}$ with the exception that
%if $c'\geq c+1$  and $x(\calC_{L_{c+0.5}})=
%x(\calC_{L_{c+1}})$, then we also remove $L^*_c$  from $S_2$.
%Further, for the lists $L$ of $S_2$ following their order in $S_2$,
%the values $x(\calC_{L})$ are strictly increasing and the values
%$\delta(\calC_{L})$ are strictly decreasing.
%The above has discussed the general case where $c\in [b+1,a-1]$. If $c=b$ or $c=a$, one can easily check that we still have $S_2=\{L_{b+1},L_{b+2},\ldots,L_{c'}\}$.
%Note that in any case $L_{b+1}$ is the first list in $S_2$.

Finally, we combine the lists of the two sets $S'_1$ and $S'_2$ to
obtain $\calL$, as follows.
%Note that the last indices of all lists of
%$S_1$ are $i$. For $S_2$, the last indices of all lists before and
%including $L_c$ are $i$ and the rest lists all have $m$ as their
%last indices. Note that $I_{i}\subseteq I_m$ in the input.
Recall that $L_{a_2}$ is the last list of $S'_1$. We consider the
lists of $S_2'$ in order. Define $b'$ to be the index $j$ of the first
list $L_j$ of $S_2'$ such that
$\delta(\calC_{L_{a_2}})>\delta(\calC_{L_j})$, and if no such list
$L_j$ exists, then let $b'=c_2+1$.

%If $\delta(\calC_{L_{a_2}})>
%\delta(\calC_{L_{b+1}})$, then let $b'=b$; otherwise, let $b'$ be the
%largest index such that
%$\delta(\calC_{L_{a_2}})\leq \delta(\calC_{L_{b'}})$ for a list $L_{b'}$
%of $S_2$.
%We have the following lemma.

\begin{lemma}\label{lem:170}
\begin{enumerate}
\item
$x(\calC_{L_{a_2}})<x(\calC_{L_{b+1}})$.
\item
If $b'>b+1$, then $L_{a_2}$ dominates $L_{j}$ for any list $L_j\in
S_2'$ with $j< b'$.
\end{enumerate}
\end{lemma}
\begin{proof}
For $L_{a_2}$, since $a_1=0$,
%recall that $x(\calC_{L_{a_2}})=\max\{x(\calC_{L'_{a_2}})+|I_{i}|, x_{i}^r\}$.
%Further, since $x_{i}^r< x_{m}^r$, it holds that
%$x(\calC_{L'_{a_2}})\geq x_{m}^r > x_{i}^l$, and thus,
%$x(\calC_{L'_{a_2}})+|I_{i}|>x_{i}^r$. Consequently,
we have $x(\calC_{L_{a_2}})=x(\calC_{L'_{a_2}})+|I_{i}|$.
For $L_{b+1}$, it holds that
$x(\calC_{L_{b+1}})=x(\calC_{L'_{b+1}})+|I_{i}|$.
%(regardless of the value of $c$).
Since  $x(\calC_{L'_{a_2}})< x(\calC_{L'_{b+1}})$,
we have $x(\calC_{L_{a_2}})<x(\calC_{L_{b+1}})$. This proves the first
statement of the lemma.

Next we prove the second lemma statement. Assume
$b'>b+1$. Consider any list $L_j\in S_2'$ with $j<b'$. In the following, we show that
$L_{a_2}$ dominates $L_j$.

Recall that the values $\delta(L)$ of the lists $L$ of $S'_2$ are strictly
decreasing following their order in $S'_2$.
By the definition of $b'$,
$\delta(\calC_{L_{a_2}})\leq \delta(\calC_{L_{j}})$.
Note that the last index of $L_j$ can be
either $i$ or $m$, and the last index of $L_{a_2}$ is $i$.

If the last index of $L_j$ is $i$, then since
$\delta(\calC_{L_{a_2}})\leq \delta(\calC_{L_{j}})$ and
$x(\calC_{L_{a_2}})<x(\calC_{L_{b+1}})\leq x(\calC_{L_{j}})$, by Lemma~\ref{lem:prune2},
$L_{a_2}$ dominates $L_j$.

If the last index of $L_j$ is $m$, then
%let $L_j'$ be the original
%list of $L_j$ before $I_{i}$ is processed. Note that if $j=c+0.5$,
%then $L_j'$ refers to $L_c'$.
%Note that
$\delta(\calC_{L_{j}})=\max\{\delta(\calC_{L'_{j}}),
d(m,\calC_{L_{j}})\}$.
Recall that $\delta(\calC_{L_{a_2}})=\max\{\delta(\calC_{L'_{a_2}}),
d(i,\calC_{L_{a_2}})\}$ and $\delta(\calC_{L'_{a_2}})>\delta(\calC_{L'_{j}})$.
Due to $\delta(\calC_{L_{a_2}})\leq \delta(\calC_{L_{j}})$, we can
deduce
$\delta(\calC_{L'_{j}})<\delta(\calC_{L'_{a_2}})\leq
\delta(\calC_{L_{a_2}})\leq \delta(\calC_{L_{j}})$.
Therefore, $\delta(\calC_{L_{a_2}})\leq \delta(\calC_{L_j})=d(m,\calC_{L_{j}})$.
Again, $x(\calC_{L_{a_2}})<x(\calC_{L_{b+1}})\leq x(\calC_{L_{j}})$.
Since the last index of $L_{a_2}$ is $i$ and that of
$L_{j}$ is $m$, with $I_{i}\subseteq I_{m}$ in the input,
by Lemma~\ref{lem:prune1}, $L_{a_2}$ dominates $L_j$.
\qed
\end{proof}

By Lemma~\ref{lem:170}, we let $\calL$ be the union of the lists of
$S'_1$ and the lists of $S'_2$ after and including $b'$ (if $b'=c_2+1$,
then $\calL=S_1'$).

\begin{observation}
All algorithm invariants hold on $\calL$.
\end{observation}
\begin{proof}
As the analysis in Section~\ref{sec:mainfirst}, $S'_1\cup S'_2$ must
contain a canonical list of $\calI[1,i]$. In light of
Lemma~\ref{lem:170}(2), $\calL$ also contains a canonical list.

Also, the values of $x(\calC_{L})$ for all lists $L$ of $S'_1$ (resp.,
$S'_2$) are strictly increasing. By Lemma~\ref{lem:170}(1), the values of
$x(\calC_{L})$ for all lists $L$ of $\calL$ are also strictly increasing.
On the other hand, the values of $\delta(\calC_{L})$ for all lists $L$ of $S'_1$ (resp.,
$S'_2$) are strictly decreasing. The definition of $b'$ makes sure
that the values of $\delta(\calC_{L})$ for all lists $L$ of $\calL$ must be  strictly
decreasing. Also, note that the lists of $\calL$ whose last indices
are $i$ are all before the lists whose last indices are $m$.

Hence, all algorithm invariants hold on $\calL$.
\qed
\end{proof}

The following lemma will be useful for the algorithm implementation.

\begin{lemma}\label{lem:case22}
For each list $L_j\in \calL$, if $L_j\neq L_c^*$, then
$x(\calC_{L_j})=x(\calC_{L'_j})+|I_{i}|$; if $L_j\not\in \{L_{a_2},L_c^*,
L_{c_1}, L_{c_2}\}$, then $\delta(\calC_{L_j})=\delta(\calC_{L'_j})$.
\end{lemma}
\begin{proof}
Consider any list $L_j\in \calL$.

If $L_j\neq L_c^*$, then since $a_1=0$,
$x(\calC_{L_j})=x(\calC_{L'_j})+|I_{i}|$ always holds regardless
whether the last index of $L_j$ is $i$ or $m$.

Assume $L_j\not\in \{L_{a_2}, L_c^*, L_{c_1}, L_{c_2}\}$. To prove
that $\delta(\calC_{L_j})=\delta(\calC_{L'_j})$, if $j\leq
b$, then we can apply the analysis in the proof of
Lemma~\ref{lem:case11}; otherwise, we can apply the analysis in the
proof of Lemma~\ref{lem:case12}. We omit the details.
\qed
\end{proof}

\paragraph{The third subcase $x_{i}^r< x_{m'}^r$.}

In this case, for each list $L_j'\in \calL'$, as analyzed in the proof
of Lemma~\ref{lem:endindex}, only Case II of our preliminary algorithm
happens, and thus $L_j$ is obtained from $L_j'$ by inserting $i$
into $L_j'$ right before the last index. Further, it holds that
$x(\calC_{L_j})=x(\calC_{L'_j})+|I_{i}|$ regardless of whether the last index of $L_j'$ is
$m$ or $m'$.
Since $x(\calC_{L'_j})$ is strictly
increasing on $j\in [1,a]$, $x(\calC_{L_j})$ is also
strictly increasing on $j\in [1,a]$.
%If $|\calL'|=1$, then
%$|\calL|=1$ and we do not need to do any pruning work. In the
%following, we assume $|\calL|>1$.

Consider any list $L'_j\in \calL'$ with $j\leq b$. Recall that
the last index of $L'_j$ is $m'$. By Observation~\ref{obser:case2},
$\delta(\calC_{L_j})=\max\{\delta(\calC_{L'_j}),d(m',\calC_{L_j})\}$,
%As discussed before, since $I_{i}\subseteq I_{m'}$ in the input,
%$d(i,\calC_{L_j})\leq d(m',\calC_{L_j})$. Hence,
%and $\delta(\calC_{L_j})=\max\{\delta(\calC_{L'_j}),d(m',\calC_{L_j})\}$.
%Note that
and
$d(m',\calC_{L_j})=x_{m'}^r(\calC_{L_j})-x_{m'}^r=x(\calC_{L_j})-x_{m'}^r$.
Thus,  $d(m',\calC_{L_j})$ strictly
increases on $j\in [1,b]$. Since  $\delta(\calC_{L'_j})$ strictly decreases on
$j\in [1,b]$, $\delta(\calC_{L_j})$ is a unimodal function on
$j\in [1,b]$ (i.e., it first strictly decreases and then strictly
increases). If $\delta(\calC_{L_{1}})\leq \delta(\calC_{L_2})$, then
let $e_1=1$; otherwise, define $e_1$ to be the largest index $j\in
[2,b]$ such that
$\delta(\calC_{L_{j-1}})>\delta(\calC_{L_j})$. Hence,
$\delta(\calC_{L_j})$ is strictly decreasing on $j\in [1,e_1]$.

\begin{lemma}\label{lem:190}
If $e_1<b$, then $L_{e_1}$ dominates $L_j$ for any $j\in [e_1+1,b]$.
\end{lemma}
\begin{proof}
Assume $e_1<b$ and let $j$ be any index in $[e_1+1,b]$. By our
definition of $e_1$ and since $\delta(\calC_{L_j})$ is unimodal on
$[1,b]$, it holds that $\delta(\calC_{L_{e_1}})\leq \delta(\calC_{L_j})$. Recall that
$x(\calC_{L_{e_1}})<x(\calC_{L_j})$. Since the last indices of both
$L_{e_1}$ and $L_j$ are $m'$, by Lemma~\ref{lem:prune2},
$L_{e_1}$ dominates $L_j$. \qed
\end{proof}

%Therefore, among the lists $L_1,L_2,\ldots,L_b$, we only need to keep
%$L_1,L_2,\ldots,L_{e_1}$.
Due to Lemma~\ref{lem:190}, let $S_1=\{L_1,L_2,\ldots,L_{e_1}\}$.

Consider any list $L'_j\in \calL'$ with $j> b$. Recall that
the last index of $L'_j$ is $m$. Similarly as above,
$\delta(\calC_{L_j})=\max\{\delta(\calC_{L'_j}),d(m,\calC_{L_j})\}$
and $d(m,\calC_{L_j})=x(\calC_{L_j})-x_{m}^r$.
%Thus, as $j$ increases in $[1,b]$, $d(m',\calC_{L_j})$ strictly
%increases. Since  $\delta(\calC_{L'_j})$ strictly decreases on
%$d(m',\calC_{L_j})$,
Similarly, $\delta(\calC_{L_j})$ is a unimodal function on
$j\in [b+1,a]$. If $\delta(\calC_{L_{b+1}})\leq \delta(\calC_{L_{b+2}})$, then we let $e_2=b+1$;
otherwise, define $e_2$ to be the largest index $j\in [b+1,a]$ such that
$\delta(\calC_{L_{j-1}})>\delta(\calC_{L_j})$. Hence,
$\delta(\calC_{L_j})$ is strictly decreasing on $j\in [b+1,e_2]$.
By a similar proof as Lemma~\ref{lem:190}, we can show that
if $e_2<a$, then $L_{e_2}$ dominates $L_j$ for any $j\in [e_2+1,a]$.
%Therefore, among the lists $L_{b+1},L_{b+2},\ldots,L_a$, we only need to keep
%$L_{b+1},L_{b+2},\ldots,L_{e_2}$.
Let $S_2=\{L_{b+1},L_{b+2},\ldots,L_{e_2}\}$.

We finally combine $S_1$ and $S_2$ to obtain $\calL$ as follows.
Define $b'$ to be the smallest index $j$ of $[b+1,e_2]$ such that
$\delta(\calC_{L_{e_1}})> \delta(\calC_{L_{j}})$, and if no such index
exists, then let $b'=e_2+1$.

\begin{lemma}\label{lem:200}
If $b'>b+1$, then $L_{e_1}$ dominates $L_{j}$ of $S_2$
for any $j\in [b+1,b'-1]$.
\end{lemma}
\begin{proof}
Assume $b'>b+1$ and let $j$ be any index in $[b+1,b'-1]$.
Since $\delta(\calC_{L_j})$ is strictly decreasing on $j\in
[b+1,e_2]$, by the definition of $b'$,
%and also due to that
%$\delta(\calC_{L_j})$ is a unimodal function on $j\in [b+1,a]$,
$\delta(\calC_{L_{e_1}})\leq \delta(\calC_{L_j})$.

Recall that
$\delta(\calC_{L_j})=\max\{\delta(\calC_{L'_j}),d(m,\calC_{L_j})\}$,
$\delta(\calC_{L_{e_1}})=\max\{\delta(\calC_{L'_{e_1}}),d(m',\calC_{L_{e_1}})\}$,
and $\delta(\calC_{L'_j})<\delta(\calC_{L'_{e_1}})$.
Hence, we obtain $\delta(\calC_{L'_j})<\delta(\calC_{L'_{e_1}})\leq
\delta(\calC_{L_j})$, and thus $\delta(\calC_{L_j})=d(m,\calC_{L_{j}})$.
Since $\delta(\calC_{L_{e_1}})\leq \delta(\calC_{L_j})$,
$\delta(\calC_{L_{e_1}})\leq d(m,\calC_{L_{j}})$.
Further, recall that $x(\calC_{L_{e_1}})< x(\calC_{L_j})$.
Then, Lemma~\ref{lem:prune1} applies since the last index of $L_{e_1}$
is $m'$ and that of $L_{j}$ is $m$, with $x_{m'}^r\leq x_{m}^r$.
By Lemma~\ref{lem:prune1}, $L_{e_1}$ dominates $L_{j}$.
%The lemma thus follows.
\qed
\end{proof}

In light of Lemma~\ref{lem:200}, we let $\calL=S_1\cup
\{L_{b'},\ldots,L_{e_2}\}$ if $b'\neq e_2+1$ and $\calL=S_1$ otherwise.
By similar analysis as before, we can show that all
algorithm invariants hold on $\calL$, and we omit the details.
The following lemma will be useful for the algorithm implementation.

\begin{lemma}\label{lem:case23}
For each list $L_j\in \calL$, $x(\calC_{L_j})=x(\calC_{L'_j})+|I_{i}|$;
if $L_j\not\in \{L_{e_1}, L_{e_2}\}$, then $\delta(\calC_{L_j})=\delta(\calC_{L'_j})$.
\end{lemma}
\begin{proof}
We have shown that $x(\calC_{L_j})=x(\calC_{L'_j})+|I_{i}|$ for
any $j\in [1,a]$.

Consider any list $L_j\in \calL$ and $j\not\in\{e_1,e_2\}$. By the
similar analysis as in Lemma~\ref{lem:case12}, we can show that
$\delta(\calC_{L_j})=\delta(\calC_{L'_j})$. The details are omitted.
\qed
\end{proof}

\subsection{The Algorithm Implementation}
\label{sec:imple}

In this section, we implement our pruning algorithm described in
Section~\ref{sec:prune} in $O(n\log n)$ time and $O(n)$ space.
We first show how to compute the optimal value $\delta_{opt}$
and then show how to construct an optimal list $L_{opt}$ in
Section~\ref{sec:optlist}.

Since $\calL$ may have $\Theta(n)$ lists and each list may have
$\Theta(n)$ intervals, to avoid $\Omega(n^2)$ time, the key idea
is to maintain the lists of $\calL$ implicitly. We show that it is
sufficient to maintain the ``$x$-values'' $x(\calC_{L})$ and the
``$\delta$-values'' $\delta(\calC_{L})$ for all lists $L$ of $\calL$,
as well as the list index $b$ and the interval indices $m'$ and $m$.
To this end, and in particular, to update the $x$-values and the $\delta$-values after each interval $I_{i}$ is processed, our implementation heavily relies  on
Lemmas~\ref{lem:case11}, \ref{lem:case12}, \ref{lem:case22}, and
\ref{lem:case23}. Intuitively, these lemmas guarantee that although the
$x$-values of all lists of $\calL$ need to change, all but a
constant number of them increase by the same amount, which can be
updated implicitly in constant time; similarly, only a constant number
of $\delta$-values need to be updated. The details are given below.

Let $\calL=\{L_1,L_2,\ldots,L_a\}$ such that $x(\calC_{L_j})$ strictly
increases on $j\in [1,a]$, and thus, $\delta(\calC_{L_j})$ strictly
decreases on $j\in [1,a]$ by the algorithm invariants.

We maintain a balanced binary search tree $T$ whose leaves from left to
right correspond to the ordered lists of $\calL$. Let $v_1,\ldots,v_a$
be the leaves of $T$ from left to right, and thus, $v_j$ corresponds
to $L_j$ for each $j\in [1,a]$. For each $j\in [1,a]$, $v_j$ stores
a {\em $\delta$-value} $\delta(v_j)$ that is equal to
$\delta(\calC_{L_j})$, and $v_j$ stores another {\em
$x$-value} $x(v_j)$ that is equal to $x(\calC_{L_j})-R$, where $R$ is
a {\em global shift} value maintained by the algorithm.

In addition, we maintain a pointer $p_b$ pointing to the leaf $v(b)$
of $T$ if $b\neq 0$ and $p_b=null$ if $b=0$.  We also maintain the
interval indices $m$ and $m'$. Again, if $p_b=null$, then $m'$ is
undefined.

Initially, after $I_1$ is processed, $\calL$ consists of the single
list $L=\{1\}$. We set $R=0$, $m=1$, and $p_b=null$. The tree $T$ consists of only one
leaf $v_1$ with $\delta(v_1)=0$ and $x(v_1)=x_1^r$.

In general, we assume $I_{i-1}$ has been processed and $T$, $m$, $m'$,
$p_b$, and $R$ have been correctly maintained. In the following, we
show how to update them for processing $I_{i}$. In particular, we
show that processing $I_{i}$ takes $O((k+1)\log n)$ time, where $k$
is the number of lists removed from $\calL$ during processing
$I_{i}$. Since our algorithm will generate at most $n$ new lists for
$\calL$ and each list will be removed from $\calL$ at most once, the total time of
the algorithm is $O(n\log n)$.

As in Section~\ref{sec:prune}, we let $\calL'=\{L_1',L_2',\ldots,L_a'\}$ denote the original
set $\calL$ before $I_{i}$ is processed.  Again, if $b\neq 0$, then
$\calL_1'=\{L_1',\ldots,L_b'\}$ and $\calL_2'=\{L_{b+1}',\ldots,L_a'\}$.
We consider the five subcases discussed in Section~\ref{sec:prune}.

\subsubsection{The Case $\calL_1'=\emptyset$}
\label{sec:implemainfirst}

In this case, the last indices of all lists of $\calL'$ are $m$.

\paragraph{The first subcase $x_{m}^r\leq x_{i}^r$.}
In this case, in general we have $\calL=\{L_j\ |\ a_1\leq j\leq a_2\}$. We first find $a_1$ and remove the lists $L_1,\ldots,L_{a_1-1}$ if $a_1>1$ as follows.

Starting from the leftmost leaf $v_1$ of $T$, if $x(v_1)+R$ (which is
equal to $x(\calC_{L_1'})$) is larger than $x_{i}^l$, then $a_1=0$ and
we are done. Otherwise, we consider the next leaf $v_2$. In general,
suppose we are considering leaf $v_j$. If $x(v_j)+R> x_{i}^l$,
then we stop with $a_1=j-1$. Otherwise, %If $x(v_j)+R\leq x_{i}^l$, then
we remove leaf $v_{j-1}$ (not $v_j$) from $T$ and continue to consider
the next leaf $v_{j+1}$ if $j\neq a$ (if $j=a$, then we stop with $a_1=a$).

If $a_1\neq 0$, then the above has found the leaf $v_{a_1}$.
In addition, we update $x(v_{a_1})=x_{i}^r-R-|I_i|$ (we have minus $|I_i|$ here because
later we will increase $R$ by $|I_i|$).

Next we find $a_2$ and remove the lists $L_{a_2+1},\ldots,L_a$ (by
removing the corresponding leaves from $T$) if $a_2<a$, as follows.
Recall that for each $j\in [a_1+1,a]$,
$\delta(\calC_{L_j})=\max\{\delta(\calC_{L'_j}),d(i,\calC_{L_j})\}$,
with $\delta(\calC_{L'_j})=\delta(v_j)$ and
$d(i,\calC_{L_j})=x_{i}^l(\calC_{L_j})-x_{i}^l=
x(\calC_{L'_j})-x_{i}^l=x(v_j)+R-x_{i}^l$.  Hence, we have
$\delta(\calC_{L_j})=\max\{\delta(v_j),x(v_j)+R-x_{i}^l\}$.

If $a_1=a$, then we have $a_2=a_1$ and we are done. Otherwise we do the following.
Starting from the rightmost leaf $v_a$ of $T$, we check whether
$\max\{\delta(v_{a-1}),x(v_{a-1})+R-x_{i}^l\}\leq
\max\{\delta(v_a),x(v_a)+R-x_{i}^l\}$. If yes, we remove $v_a$ from
$T$ and continue to consider $v_{a-1}$. In general, suppose we are
considering $v_j$. If $j=a_1$, then we stop with $a_2=a_1$. Otherwise,
we check whether $\max\{\delta(v_{j-1}),x(v_{j-1})+R-x_{i}^l\}\leq
\max\{\delta(v_j),x(v_j)+R-x_{i}^l\}$. If yes, we remove $v_j$ from
$T$ and proceed on $v_{j-1}$. Otherwise, we stop with $a_2=j$.

Suppose the above procedure finds leaf $v_j$ with $a_2=j$. We further
update  $\delta(v_j)=\max\{\delta(v_j),x(v_j)+R-x_{i}^l\}$. By
Lemma~\ref{lem:case11}, we do not need to update other
$\delta$-values.

The above has updated the tree $T$. In addition, we update
$R=R+|I_{i}|$, which actually implicitly updates all $x$-values
by Lemma~\ref{lem:case11}. Finally, we update $m=i$ since the last
indices of all updated lists of $\calL$ are now $i$.

This finishes our algorithm for processing $I_{i}$.
Clearly, the total time is $O((k+1)\log n)$ since removing each leaf
of $T$ takes $O(\log n)$ time, where $k$ is the number of leaves
that have been removed from $T$.

\paragraph{The second subcase $x_{m}^r> x_{i}^r$.}

In this case, roughly speaking, we should compute the set $\calL=\{L_1,\ldots,L_{c_1},L_{c'},L_{c'+1},\ldots,L_{c_2}\}$.

We first compute the index $c$, i.e., find the leaf $v_c$ of $T$. This can be done
by searching $T$ in $O(\log n)$ time as follows. Note that for a list $L_j'$, to check
whether $x_{i}^l> x_m^l(\calC_{L'_j})$, since
$x_m^l(\calC_{L'_j})=x(\calC_{L'_j})-|I_m|=x(v_j)+R-|I_m|$, it
is equivalent to checking whether $x_{i}^l> x(v_j)+R-|I_m|$, which is
equivalent to  $x_{i}^l-R+|I_m|> x(v_j)$. Consequently, $v_c$
is the rightmost leaf $v$ of $T$ such that $x_{i}^l-R+|I_m|>
x(v)$, and thus $v_c$ can be found by searching $T$ in $O(\log n)$ time.

Next, we find $c_1$, and remove the leaves $v_j$ with $j\in [c_1+1,c]$
if $c_1<c$, as follows (note that if the above step finds $c=0$, then
we skip this step).

Recall that for each $j\in [1,c]$,
$\delta(\calC_{L_j})=\max\{\delta(\calC_{L'_j}),d(i,\calC_{L_j})\}$,
with $\delta(\calC_{L'_j})=\delta(v_j)$ and
$d(i,\calC_{L_j})=x_{i}^l(\calC_{L_j})-x_{i}^l=x(\calC_{L'_j})-x_{i}^l
=x(v_j)+R-x_{i}^l$. Hence, we have
$\delta(\calC_{L_j})=\max\{\delta(v_j),x(v_j)+R-x_{i}^l\}$.

Starting from $v_c$, we first check whether
$\delta(\calC_{L_{c-1}})>\delta(\calC_{L_c})$, by computing
$\delta(\calC_{L_{c-1}})$ and $\delta(\calC_{L_c})$ as above.
If yes, then $c_1=c$ and we stop. Otherwise, we remove $v_c$ and proceed on considering
$v_{c-1}$.  In general, suppose we are considering $v_j$. If $j=1$,
then we stop with $c_1=1$. Otherwise, we check
whether $\delta(\calC_{L_{j-1}})>\delta(\calC_{L_j})$. If yes, then
$c_1=j$; otherwise, we remove $v_j$ and proceed on $v_{j-1}$.

In addition, after $v_{c_1}$ is found as above, we update
$\delta(v_{c_1})=\max\{\delta(v_{c_1}),x(v_{c_1})+R-x_{i}^l\}$.

%(if the above removes $v_c$, then we temporarily keep the values of
%$x(v_c)$ and $\delta(v_c)$ for the possible inserting of the node for the new list
%$L^*_c$)

Next, consider the new list $L^*_c$, which is $L_{c+0.5}$. We have
$\delta(\calC_{L_c^*})= \max\{\delta(\calC_{L_c'}),d(m,\calC_{L^*_c})\}=
\max\{\delta(\calC_{L_c'}),x_{m}^l(\calC_{L^*_c})-x_m^l\}$. Since
$\delta(\calC_{L_c'})=\delta(v_c)$ and
$x_{m}^l(\calC_{L^*_c})=x_{i}^r$, we have
$\delta(\calC_{L_c^*})= \max\{\delta(v_c),x_{i}^r-x_m^l\}$ (if the above has removed $v_c$,
then we temporarily keep the value $\delta(v_c)$ before $v_c$ is removed). Also, recall that
$x(\calC_{L_c^*})=x_{i}^r+|I_m|$. Therefore, we can compute both
$\delta(\calC_{L_c^*})$ and $x(\calC_{L_c^*})$ in constant time.
We insert a new leaf $v_{c+0.5}$ to $T$ corresponding to $L^*_c$, with
$\delta(v_{c+0.5})=\delta(\calC_{L_c^*})$ and
$x(v_{c+0.5})=x(\calC_{L_c^*})-R-|I_{i}|$ (the minus $|I_i|$
is due to that later we will increase $R$ by $|I_i|$).

Next, we determine $c_2$, and remove the leaves $v_j$ with $j\in [c_2+1,a]$
if $c_2<a$, as follows. Recall that for each $j\in [c+1,a]$,
$\delta(\calC_{L_j})=\max\{\delta(\calC_{L'_j}),d(m,\calC_{L_j})\}$,
with $\delta(\calC_{L'_j})=\delta(v_j)$ and
$d(m,\calC_{L_j})=x_{m}^r(\calC_{L_j})-x_{m}^r=x(\calC_{L'_j})+|I_{i}|-x_{m}^r
=x(v_j)+R+|I_{i}|-x_{m}^r$. Hence, we have
$\delta(\calC_{L_j})=\max\{\delta(v_j),x(v_j)+R+|I_{i}|-x_{m}^r\}$,
which can be computed in constant time once we access the leaf $v_j$.

Starting from the rightmost leaf $v_a$, in general, suppose we are
considering a leaf $v_j$. If $j=c+0.5$, then we stop with $c_2=c+0.5$.
Otherwise, let $v_h$ be the left neighboring leaf of $v_j$ (so $h$ is
either $j-1$ or $j-0.5$). We check whether $\delta(\calC_{L_h})>
\delta(\calC_{L_j})$ (the two values can be computed as above).
If yes, we stop with $c_2=j$; otherwise, we remove $v_j$ from $T$
and proceed on considering $v_h$.

If the above procedure returns $c_2\geq c+1$, then we further check
whether $x(\calC_{L_{c}^*})=x(\calC_{L_{c+1}})$. If yes, then we
remove the leaf $v_{c+0.5}$ from $T$. If $c_2\geq c+1$, we also need to
update %$\delta(v_{c_2})=\max\{\delta(v_{c_2}),x(v_j)+R-x_{i}^l\}$.
$\delta(v_{c_2})=\max\{\delta(v_{c_2}),x(v_{c_2})+R+|I_{i}|-x_{m}^r\}$.

Finally, we determine $c'$ and remove all leaves strictly
between $v_{c_1}$ and $v_{c'}$, as follows.
Recall that given any leaf $v_j$ of $T$, we can compute
$\delta(\calC_{L_j})$ in constant time.
Starting from the right neighboring leaf of
$v_{c_1}$, in general, suppose we are considering a leaf
$v_j$. If $\delta(\calC_{L_{c_1}})\leq \delta(\calC_{L_j})$, then we
remove $v_j$ and proceed on the right neighboring leaf of $v_j$. This
procedure continues until either $\delta(\calC_{L_{c_1}}) >
\delta(\calC_{L_j})$ or $v_j$ is the rightmost leaf and has been
removed.

In addition, we update $R=R+|I_{i}|$. In light of
Lemma~\ref{lem:case12} and by our way of setting the
value $x(v_{c+0.5})$, this updates all $x$-values.
Also, the above has ``manually'' set the values $\delta(v_{c_1})$,
$\delta(v_{c_2})$, and $\delta(v_{c_{c+0.5}})$, by Lemma~\ref{lem:case12},
all $\delta$-values have been updated.
Finally, we update $m$, $m'$, and $p_b$ as follows.

In the general case where $1\leq c<a$ and $c'\neq c_2+1$, we set
$m'=i$ and $p_b$ to the leaf $v_{c_1}$. If $c'=c_2+1$, then the last
indices of all lists of $\calL$ are $i$, and thus we set $m=i$ and
$p_b=null$.
If $c=0$, then the last indices of all lists of $\calL$ are $m$, then
we do not need to update anything.
If $c=a$, then if $L_c^*\not\in \calL$, then
the last indices of all lists of
$\calL$ are $i$ and we set $m=i$ and $p_b=null$, and if $L_c^*\in
\calL$, then we set $m'=i$ and $p_b$ to $v_{c_1}$.

This finishes processing $I_{i}$. The total time is again as claimed
before.

\subsubsection{The Case $\calL_1'\neq \emptyset$}

In this case, $\calL'_1=\{L_1',\ldots,L_b'\}$ and
$\calL'_2=\{L_{b+1}',\ldots,L_a'\}$. The last indices of all lists of
$\calL_1'$ (resp., $\calL_2'$) are $m'$ (resp., $m$).
Note that the pointer $p_b$ points to the leaf $v_b$.

\paragraph{The first subcase $x_{i}^r\geq x_m^r$.}
In this case, the implementation is similar to the first subcase of
Section~\ref{sec:implemainfirst}, so we omit the details.

\paragraph{The second subcase $x_{m'}^r\leq x_{i}^r< x_m^r$.}

As our algorithm description in Section~\ref{sec:mainsecond}, we first
apply the similar implementation as the first subcase of
Section~\ref{sec:implemainfirst} on the leaves from $v_1$ to $v_b$, and then
apply the similar implementation as the second subcase of
Section~\ref{sec:implemainfirst} on the leaves from $v_{b+1}$ to
$v_a$. So the leaves of the current tree corresponding to the lists in
$S_1'\cup S_2'$, i.e.,
$\{L_1\ldots,L_{a_2},L_{b+1},\ldots,L_{c_1},L_{c'},\ldots,L_{c_2}\}$,
as defined in the second subcase of Section~\ref{sec:mainsecond}.

Next, we determine $b'$ and remove all leaves from $T$ strictly between
$v_{a_2}$ and $v_{b'}$. Starting from the right neighboring leaf of
$v_{a_2}$, in general, suppose we are considering a leaf
$v_j$. If $\delta(\calC_{L_{a_2}})\leq \delta(\calC_{L_j})$ (as before, these two
values can be computed in constant time once we have access to
$v_{a_2}$ and $v_j$), then we
remove $v_j$ and proceed on the right neighboring leaf of $v_j$. This
procedure continues until either $\delta(\calC_{L_{a_2}}) >
\delta(\calC_{L_j})$ or $v_j$ is the rightmost leaf and has been
removed.

Finally, we update $R=R+|I_{i}|$. To update $p_b$, $m$, and $m'$,
depending on the values $c, c'$ and $b'$, there are various cases. In the
general case where $b+1\leq c<a$, $c'\neq c_2+1$, and $b'\neq c_2+1$,
we update $p_b=v_{c_1}$ and $m'=i$. We omit the discussions for other special cases.

\paragraph{The third subcase $x_{i}^r< x_{m'}^r$.}

In this case, starting from $v_b$, we first remove all leaves from
$v_{e_1+1}$ to $v_b$. The algorithm is very similar as before and we
omit the details. Then, starting from $v_a$, we remove all leaves from
$v_{e_2+1}$ to $v_a$. Finally, starting from
$v_{e_1}$, we remove all leaves strictly between $v_{e_1}$ to
$v_{b'}$. In addition, we update $R=R+|I_{i}|$. In the general
case where $b'\neq e_2+1$, we set $p_b$ pointing to leaf $v_{e_1}$; otherwise,
we set $m=m'$ and $p_b=null$.

This finishes processing $I_{i}$ for all five subcases.
The algorithm finishes once $I_n$ is processed, after which
$\delta_{opt}=\delta(v)$, where $v$ is the rightmost leaf of $T$ (as
$\delta(v)$ is the smallest among all leaves of $T$). Again, the total
time of the algorithm is $O(n\log n)$. Clearly, the space used by our algorithm is $O(n)$.

\subsection{Computing an Optimal List}
\label{sec:optlist}
As discussed above, after $I_n$ is processed, the list (denoted by $L_{opt}$) corresponding
to the rightmost leaf (denoted by $v_{opt}$) of $T$ is an optimal
list, and $\delta_{opt}=\delta(v_{opt})$.
%we know that if $L$ is the list corresponding to the
%rightmost leaf $v$ of $T$, then $L$ is an optimal list and
%$\delta_{opt}=\delta(v)$. Hence, $\delta_{opt}$ is computed.
However, since our algorithm does not maintain the list $L_{opt}$
explicitly, $L_{opt}$ is not available after the algorithm finishes.
%(indeed we only know the last
%index of $L$ by the information $m'$, $m$, and $p_b$).
In this section, we give a way (without changing the complexity asymptotically) to maintain more information during the
algorithm such that after it finishes, we can reconstruct $L_{opt}$ in
additional $O(n)$ time.
%In the following, we use $L_{opt}$ to denote the above list $L$.

We first discuss some intuition. Consider a list $L\in \calL$ before interval
$I_{i}$ is processed. During processing $I_{i}$ for $L$, observe
that the position of $i$ in the updated list $L$ is uniquely
determined by the input position of the last
interval $I_m$ of $L$ (i.e., depending on whether $x_{i}^r\geq
x_m^r$). However, uncertainty happens when $L$ generates
another ``new'' list $L^*$. More specifically, suppose $L$ is a canonical
list of $\calI[1,i-1]$. If there is no new list $L^*$, then by our
observations (i.e., Lemmas~\ref{lem:case1} and \ref{lem:case2}), the
updated $L$ is a canonical list of $\calI[1,i]$. Otherwise, we know (by
Lemma~\ref{lem:case3}) that one of $L$
and $L^*$ is a canonical list of $\calI[1,i]$, but we do not know
exactly which one is. This
is where the uncertainty happens and indeed this is why we
need to keep both $L$ and $L^*$ (thanks to Lemma~\ref{lem:prune}, we
only need to keep one such new list). Therefore, in order to reconstruct
$L_{opt}$, if processing $I_{i}$ generates a new list $L^*$ in
$\calL$, then we need to keep the relevant information about $L^*$.
The details are given below.

Specifically, we maintain an additional binary tree $T'$ (not a search tree). As in $T$, the leaves of $T'$ from left to right
correspond to the ordered lists of $\calL$. Consider a leaf $v$ of
$T'$ that corresponds to a list $L\in \calL$. Suppose after
processing $I_{i}$, $L$ generates a new list $L^*$ in $\calL$.
Let $m$ be the last index of the original $L$ (before $I_i$ is processed).
According to our algorithm, we know
that the last two indices of the updated $L$ are $m$ and $i$ with $i$ as the last index and
the last two indices of $L^*$ are $i$ and $m$ with $m$ as the last index.
Correspondingly, we update the
tree $T'$ as follows. First, we store $i$ at $v$, e.g., by setting
$A(v)=i$, which means that there are two choices for processing
$I_{i}$. Second, we create two children $v_1$ and $v_2$
for $v$ and they correspond to the lists $L$ and $L^*$, respectively.
Thus, $v$ now becomes an internal node.
Third, on the new edge $(v,v_1)$, we store an ordered pair $(m,i)$,
meaning that $m$ is before $i$ in $L$; similarly, on the edge
$(v,v_2)$, we store the pair $(i,m)$. In this way, each internal node of $T'$
stores an interval index and each edge of $T'$ stores an ordered pair.

After the algorithm finishes, we reconstruct the list $L_{opt}$ in the
following way. Let $\pi$ be the path from the root to the rightmost leaf $v_{opt}$
of $T'$. We will construct $L_{opt}$ by considering
all intervals  from $I_1$ to $I_{n}$
and simultaneously considering the nodes in $\pi$.
Initially, let $L_{opt}=\{1\}$.
%and the configuration $\calC_L$ is also computed.
Then, we consider $I_2$
and the first node of $\pi$ (i.e., the root of $T'$).
In general, suppose we are considering $I_{i}$ and a node $v$ of
$\pi$. We first assume that $v$ is an internal node (i.e., $v\neq v_{opt}$).

If $i<A(v)$, then only Case I or Case II of our preliminary algorithm
happens, and we insert $i$ into $L_{opt}$ based on whether $x_{i}^r\geq
x_m^r$ (specifically, if $x_{i}^r\geq
x_m^r$, then we append $i$ at the end of $L_{opt}$; otherwise, we insert $i$ right before
the last index of $L_{opt}$) and then proceed on $I_{i+1}$.

If $i\geq A(v)$ (in fact, $i$ must be equal to $A(v)$), then we insert $i$
into $L_{opt}$ based on the ordered pair of the
next edge of $v$ in $\pi$ (specifically, if $i$ is at the second
position of the pair, then $i$ is appended at the end of $L_{opt}$;
otherwise, $i$ is inserted right before the last index of $L_{opt}$) and
then proceed on the next node of $\pi$ and $I_{i+1}$.

If $v=v_{opt}$, then we
insert $i$ into $L_{opt}$ based on whether $x_{i}^r\geq x_m^r$ as above, and then proceed on $I_{i+1}$.
The algorithm finishes once $I_{n}$ is processed, after which
$L_{opt}$ is constructed.  It is easy to see that the algorithm
runs in $O(n)$ time and $O(n)$ space.

Once $L_{opt}$ is computed, %by Lemma~\ref{lem:left},
we can apply the left-possible placement
strategy to compute an optimal configuration in additional $O(n)$ time.

\begin{theorem}
Given a set of $n$ intervals on a line, the interval separation
problem is solvable in $O(n\log n)$ time and $O(n)$ space.
\end{theorem}

\section{Conclusions}
\label{sec:conclude}

In this paper, we present an $O(n\log n)$ time and $O(n)$ space algorithm for solving the interval separation problem.
By a linear-time reduction from the integer element distinctness problem \cite{ref:LubiwA91,ref:YaoLo91}, we
can obtain an $\Omega(n\log n)$ time lower bound for the problem under the algebraic decision
tree model, which implies the optimality of our algorithm.

Given a set of $n$ integers $A=\{a_1,a_2,\ldots,a_n\}$, the element distinctness problem is to ask
whether there are two elements of $A$ that are equal. The problem has an $\Omega(n\log n)$ time lower bound
under the algebraic decision tree model \cite{ref:LubiwA91,ref:YaoLo91}. We create a set $\calI$ of
$n$ intervals as an
instance of our interval separation problem as follows. For each $a_i\in A$, we create an
interval $I_i$ centered at $a_i$ with length $0.1$. Let $\calI$ be the set of all intervals.
Since all elements of $A$ are integers,
it is easy to see that no two elements of $A$ are equal if and only if
no two intervals of $\calI$ intersect. On the other hand, no two intervals of $\calI$ intersect
if and only if the optimal value $\delta_{opt}$ in our interval separation problem on $\calI$
is equal to zero. This completes the reduction. This reduction actually shows that even if all
intervals have the same length, the interval separation problem still has an $\Omega(n\log n)$ time lower bound.

%============= below is the bibliography information =================
\bibliography{reference}
\bibliographystyle{plain}

\end{document}